\documentclass[12pt]{article}

\usepackage{mathtools}
\usepackage{booktabs}
\usepackage[english]{babel} % English language/hyphenation
\usepackage[protrusion=true,expansion=true]{microtype} % Better typography
\usepackage{amsmath,amsfonts,amsthm}
\usepackage{ amssymb }
\usepackage{graphicx}
\usepackage{array}
\usepackage[toc,page]{appendix}
\usepackage{xr}
\usepackage{booktabs} % Horizontal rules in tables
\usepackage{natbib}
\usepackage{setspace}
\usepackage{microtype} % Slightly tweak font spacing for aesthetics
\usepackage{tabularx}
\usepackage{subcaption}
\usepackage[font = small,labelfont=bf,textfont=it]{caption} % Custom captions under/above floats in tables or figures
\usepackage{footnote}
\usepackage{algorithm}% http://ctan.org/pkg/algorithms
\usepackage[noend]{algpseudocode}% http://ctan.org/pkg/algorithmicx
\usepackage{multirow}

\makeatletter
\newcommand{\distas}[1]{\mathbin{\overset{#1}{\kern\z@\sim}}}%
\newsavebox{\mybox}\newsavebox{\mysim}
\newtheorem{theorem}{Theorem}[section]

\newtheorem{lemma}[theorem]{Lemma}
\theoremstyle{definition}
\newtheorem{definition}{Definition}[section]

\newtheorem{remark}[theorem]{Remark}
\newcommand{\distras}[1]{%
  \savebox{\mybox}{\hbox{\kern3pt$\scriptstyle#1$\kern3pt}}%
  \savebox{\mysim}{\hbox{$\sim$}}%
  \mathbin{\overset{#1}{\kern\z@\resizebox{\wd\mybox}{\ht\mysim}{$\sim$}}}%
}
\newcolumntype{C}[1]{>{\centering\let\newline\\\arraybackslash\hspace{0pt}}m{#1}}

\newcommand{\blind}{1}

\begin{document}

\def\spacingset#1{\renewcommand{\baselinestretch}%
{#1}\small\normalsize} \spacingset{1.3}

%%%%%%%%%%%%%%%%%%%%%%%%%%%%%%%%%%%%%%%%%%%%%%%%%%%%%%%%%%%%%%%%%%%%%%%%%%%%%%

\if1\blind
{
  %\title{\bf A Generalized Gaussian Process Model for Computer Experiments with Binary Time Series}
	\centering{\bf\Large A Generalized Gaussian Process Model for Computer Experiments with Binary Time Series}\\
	  \vspace{0.2in}
	  \centering{Chih-Li Sung$^{a}$ $^{1}$, Ying Hung$^{b}$ \footnote{Joint first authors.},
	  %\footnote{These authors contributed equally to the manuscript.}$^{,b}$, 
	  William Rittase$^c$, Cheng Zhu$^c$,\\ C. F. J. Wu$^{a}$ \footnote{Corresponding author.} }\\
	      \vspace{0.2in}
	  \centering{$^a$School of Industrial and Systems Engineering, Georgia Institute of Technology\\
	  $^b$Department of Statistics, Rutgers, the State University of New Jersey\\
	  $^c$Department of Biomedical Engineering, Georgia Institute of Technology\\}
	  %\maketitle
} \fi

\if0\blind
{
  \bigskip
  \bigskip
  \bigskip
  \begin{center}
    {\LARGE\bf A Generalized Gaussian Process Model for Computer Experiments with Binary Time Series}
\end{center}
  \medskip
} \fi

\bigskip
\begin{abstract}

Non-Gaussian observations such as binary responses are common in some computer experiments. Motivated by the analysis of a class of cell adhesion experiments, we introduce a generalized Gaussian process model for binary responses, which shares some common features with standard GP models. In addition, the proposed model incorporates a flexible mean function that can capture different types of time series structures. Asymptotic properties of the estimators are derived, and an optimal predictor as well as its predictive distribution are constructed. Their performance is examined via two simulation studies. The methodology is applied to study computer simulations for cell adhesion experiments. The fitted model reveals important biological information in repeated cell bindings, which is not directly observable in lab experiments.
 
\end{abstract}

\noindent%
\vfill

\newpage
\spacingset{1.45} % DON'T change the spacing!

\section{Introduction}

Cell adhesion plays an important role in many physiological and pathological processes.
This research is motivated by the analysis of a class of cell adhesion experiments called micropipette adhesion frequency assays, which is a method for measuring the kinetic rates between molecules in their native membrane environment. 
%Figure 1 illustrates the setting of the micropipette adhesion frequency assay \citep{Chesla1998}. 
In a micropipette adhesion frequency assay, a red blood coated in a specific ligand is brought into contact with cell containing the native receptor for a predetermined duration, then retracted. The output of interest is binary, indicating whether a controlled contact results in adhesion. If there is an adhesion between molecules at the end of contact, retraction will stretch the red cell. If no adhesion resulted, the red cell will not be stretched. The kinetics of the molecular interaction can be derived through many repeated trials. In theory, these contacts should be independent Bernoulli trials.
However, there is a memory effect in the repeated tests and the quantification of such a memory effect is scientifically important \citep{Zarnitsyna2007PNAS, hung2012binary}.  
 
%Cell adhesion plays an important role in many physiological and pathological processes. This research is motivated by the analysis of a class of cell adhesion experiments called repeated adhesion frequency assays, which is the only published method for measuring the kinetic rates of cell adhesion.  Figure 1 illustrates the setting of the micropipette adhesion frequency assay \citep{Chesla1998}. Here, a human red blood cell (Figure 1, left) pressurized by micropipette suction is used as a force transducer to test interactions between molecules presented on the red cell membrane and the counter molecules on the surface of another cell (Figure 1, right). The two cells are put together for a pre-determined duration (Figure 1B), then retracted away. The output of interest here is binary, indicating whether a controlled contact results in adhesion. If adhesion is resulted, retraction will stretch the red cell (Figure 1C). If no adhesion is resulted, the red cell will not be stretched (Figure 1A). It was discovered that there is a memory effect in the repeated tests and the quantification of such a memory effect is scientifically important \citep{Zarnitsyna2007PNAS, hung2012binary}. 

%\begin{figure}[h]
%\centering\resizebox{280pt}{90pt} {\includegraphics{exp3fig.pdf}}
%\caption{Photomicrographs of micropipette adhesion frequency
%assay}\label{y:exp}
%\end{figure}

A cost-effective way to study the repeated adhesion frequency assays is through computer experiments, which study real systems using complex mathematical models and numerical tools such as finite element analysis \citep{santner2003design}. They have been widely used as alternatives to physical experiments or observations, especially for the study of complex systems. %In many such situations, a physical experiment or observation may be infeasible because it is unethical, impossible, inconvenient or too expensive. 
For cell adhesion, performing physical experiments (i.e., lab work) is time-consuming and often involves complicated experimental manipulation. Therefore, instead of performing the experiments only based on the actual lab work, computer simulations based on mathematical models are conducted to provide an efficient way to examine the complex mechanisms behind the adhesion.

The analysis of computer experiments has three objectives: (i) to build a model that captures the nonlinear relationship between inputs and outputs; (ii) to estimate the unknown parameters in the model and deduce properties of the estimators; (iii) to provide an optimal predictor for untried input settings, also called ``emulator'' or ``surrogate model'', and quantify its predictive uncertainty \citep{sacks89, santner2003design}. This objective (iii) is crucial because computer simulations are generally expensive or time-consuming to perform and therefore the emulators based on computer simulations are used as surrogates to perform sensitivity analysis, process optimization, calibration, etc. In particular, it is critical for calibration problems in which the emulators and physical experiments are integrated so that some unknown calibration parameters can be estimated. In the literature, Gaussian process (GP) model, use of which achieves the three objectives, is widely used for the analysis of computer experiments. A GP model accommodates nonlinearity using GP and provides an optimal predictor with an interpolation property. The applications of GP can be found in many fields in science and engineering. %Beyond deterministic simulations, GP-like models are also successful in applications to stochastic simulations by incorporating a nugget effect into the model (\cite{santner2003design}, Section 4).

The conventional GP models are developed for continuous outputs with a Gaussian assumption, which does not hold in some scientific studies. For example, the focus of the cell adhesion frequency assays is to elicit the relationship between the setting of kinetic parameters/covariates and the adhesion score, which is binary. For binary outputs, the Gaussian assumption is not valid and GP models cannot be directly applied. Binary outputs are common in computer experiments, but the extensions of GP models to non-Gaussian cases have received scant attention in computer experiment literature. Although there are intensive studies of generalized GP models for non-Gaussian data in machine learning and spatial statistics literature, such as \cite{williams1998bayesian}, \cite{zhang2002estimation}, \cite{rasmussen2006gaussian}, \cite{nickisch2008approximations} and \cite{wang2014generalized}, the asymptotic properties of estimators have not been systematically studied. Moreover, an analogy to the GP predictive distribution for binary data is important for uncertainty quantification in computer experiments, which has not yet been developed to the best of our knowledge.

Apart from the non-Gaussian responses, analysis of the repeated cell adhesion frequency assays poses another challenge, namely, how to incorporate a time series structure with complex interaction effects. It was discovered that cells appear to have the ability to remember the previous adhesion events and such a memory has an impact on the future adhesion behaviors \citep{Zarnitsyna2007PNAS, hung2012binary}. 
The quantification of the memory effect and how it interacts with the settings of the kinetic parameters in the binary time series are important but cannot be obtained by direct application of the conventional GP models.
To consider the time series structure, a common practice is to construct a spatial-temporal model. However, a separable correlation function (e.g., \cite{gelfand2004nonstationary,conti2010bayesian}) in which space and time are assumed to be independent is often implemented as a convenient way to address the computational issue. As a result, the estimation of interaction between space and time, which is of major interest here, is not allowed for. Even in the cases where nonseparable correlation functions (e.g., \cite{gelfand2004nonstationary,fricker2013multivariate}) are implemented, the interaction effect is still not easily interpretable. Therefore, a new model that can model binary time series and capture interaction effects is called for.

%Second, the standard GP formulation does not consider time series structure and the interaction between time series history and the setting of covariates in the mean function. Instead, most of the conventional GP model incorporate the time element by a separable correlation function (references) in which spatial and temporal information are assumed to be independent. The separable assumption is convenient but not always appropriate. For example, in this study, it is discovered that there is a memory effect in the repeated binding and it is important to quantify the memory and understand how the kinetic setting interact with the binding events. There are some works in the literature relaxing the independent assumption, but they are mostly problem-specific (more references and discussions). 

To achieve the objectives in the analysis of computer experiments and overcome the aforementioned limitations with binary time series outputs, we introduce a new class of models in this article. The idea is to generalize GP models to non-Gaussian responses and incorporate a flexible mean function that can estimate the time series structure and its interaction with the input variables. In particular, we focus on binary responses and introduce a new model which is analogous to the GP model with an optimal interpolating predictor. Rigorous studies of estimation, prediction, and inference are required for the proposed model and the derivations are complicated by the nature of binary responses and the dependency of time series. Since binary responses with serial correlations can be observed in computer experiments, the proposed method can be readily applicable to other fields beyond cell biology. For example, in manufacturing industry computer simulations are often conducted for the failure analysis where the outputs of interest are binary, i.e., failure or success \citep{yan2009reliability}. Examples can also be found in other biological problems where binary outputs are observed and evolve in time, such as neuron firing simulations, cell signaling pathways, gene transcription, and recurring diseases \citep{gerstner1997neural,mayrhofer2002devs}. 
The proposed method can also be broadly applied beyond computer experiments. In many scientific experiments, such as medical research and social studies, binary repeated measurements are commonly observed with serial correlations. In these situations, the proposed method can be implemented to provide a flexible nonlinear model that quantifies the correlation structure and explains the complex relationship between inputs and binary outputs. More examples can be found in functional data analysis, longitudinal data analysis, and machine learning.

The remainder of this article is organized as follows. The new class of models is discussed in Section 2. In Section 3 and 4,
asymptotic properties of the estimators are derived and the predictive distributions are constructed. Finite sample performance is demonstrated by simulations in Section 5. In Section 6, the proposed method is illustrated with the analysis of computer experiments for cell adhesion frequency assays. Concluding remarks are given in Section 7. Mathematical proofs and algorithms are provided in the online supplement. An implementation for our method can be found in \texttt{binaryGP} \citep{CL2017} in \texttt{R} \citep{R}.
 
\section{Model}
\subsection{Generalized Gaussian process models for binary response}\label{sec:simplemodel}
We first introduce a model for binary responses in computer experiments which is analogous to the conventional GP models for continuous outputs. Suppose a computer experiment has a $d$-dimensional input setting $\mathbf{x}=(x_1,\ldots,x_d)'$ and for each setting the binary output is denoted by $y(\mathbf{x})$ and randomly generated from a Bernoulli distribution with probability $p(\mathbf{x})$. Using a logistic link function, the Gaussian process model for binary data can be written as
\begin{equation}\label{eq:simplemodel}
\text{logit}(p(\mathbf{x}))=\alpha_0+\mathbf{x}'\boldsymbol{\alpha}+Z(\mathbf{x}),
\end{equation}
where $p(\mathbf{x})=\mathbb{E}[y(\mathbf{x})]$, $\alpha_0$ and $\boldsymbol{\alpha}=(\alpha_1,\ldots,\alpha_d)'$ are the intercept and linear effects of the mean function of $p(\mathbf{x})$, and $Z(\cdot)$ is a zero mean Gaussian process with variance $\sigma^2$, correlation function $R_{\boldsymbol{\theta}}(\cdot,\cdot)$, and unknown correlation parameters $\boldsymbol{\theta}$. 

Various choices of correlation functions have been discussed in the literature. For example, the \textit{power exponential correlation function} is commonly used in the analysis of computer experiments \citep{santner2003design}:
\begin{equation}\label{eq:powercorrelationfunction}
R_{\boldsymbol{\theta}}(\mathbf{x}_i,\mathbf{x}_j)=\exp\left\{-\sum^d_{l=1}\frac{(x_{il}-x_{jl})^p}{\theta_l}\right\},
\end{equation}
where $\boldsymbol{\theta}=(\theta_1,\ldots,\theta_d)$, the power $p$ controls the smoothness of the output surface, and the parameter $\theta_l$ controls the decay of correlation with respect to the distance between $x_{il}$ and $x_{jl}$. Recent studies have shown that a careful selection of the correlation function, such as orthogonal Gaussian processes proposed by \cite{plumlee2016orthogonal}, can resolve the identifiability issue in the estimation of Gaussian process models \citep{hodges2010adding,paciorek2010importance,tuo2015calibration}. 
This is particularly important in the application of calibration problems where the parameter estimation plays a significant role. 
Depending on the objectives of the studies, different correlation functions can be incorporated into the proposed model and the  theoretical results developed herein remain valid.  

Similar extensions of GP models to binary outputs have been applied in many different fields. For example, when $\mathbf{x}$ represents a two-dimensional spatial domain, (\ref{eq:simplemodel}) becomes the spatial generalized linear mixed model proposed by \cite{zhang2002estimation}. In a Bayesian framework, Gaussian process priors are implemented for classification problems, such as in \cite{williams1998bayesian} and \cite{gramacy2011particle}. Despite successful applications of these models, theoretical studies on the estimation and prediction properties are not available. Therefore, one focus of this paper is to provide theoretical supports for the estimation and prediction in (\ref{eq:simplemodel}).

\subsection{Generalized Gaussian process models for binary time series}

In this section, we introduce a new model for the analysis of computer experiments with binary time series, which is an extension of (\ref{eq:simplemodel}) that takes serial correlations between binary observations into account.
Suppose for each setting of a computer experiment, a sequence of \textit{binary time series} outputs $\{y_t(\mathbf{x})\}^T_{t=1}$ is randomly generated from Bernoulli distributions with probabilities $\{p_t(\mathbf{x})\}^T_{t=1}$. A generalized Gaussian process model for binary time series can be written as:
\begin{equation}\label{eq:model_GPBiTS}
\text{logit}(p_{t}(\mathbf{x}))=\eta_{t}(\mathbf{x})=\sum^R_{r=1}\varphi_ry_{t-r}(\mathbf{x})+\alpha_0+\mathbf{x}'\boldsymbol{\alpha}+\sum^L_{l=1}\boldsymbol{\gamma}_l\textbf{x}y_{t-l}(\mathbf{x})+Z_t(\mathbf{x}),
\end{equation}
%\begin{equation}\label{eq:model_GPBiTS}
%\text{logit}(p_{t}(\mathbf{x}))=\eta_{t}=\sum^R_{r=1}\varphi_ry_{t-r}+\sum^Q_{q=1}\zeta_q(y_{t-q}-p_{t-q})+\mathbf{x}'\boldsymbol{\alpha}+Z(\mathbf{x})+\sum^L_{l=1}\gamma_l\textbf{x}y_{t-l},
%\end{equation}
where $p_t(\mathbf{x})=\mathbb{E}[y_{t}(\mathbf{x})|H_{t}]$ is the conditional mean given the previous information $H_{t}=\{y_{t-1}(\mathbf{x}),y_{t-2}(\mathbf{x}),\ldots\}$.
In model (\ref{eq:model_GPBiTS}), $\{\varphi_r\}^R_{r=1}$ represents an autoregressive (AR) process with order $R$ and $\boldsymbol{\alpha}=(\alpha_1,\ldots,\alpha_d)'$ represents the effects of $\mathbf{x}$. The $d$-dimensional vector $\boldsymbol{\gamma}_l$ represents the interaction between the input and the past outputs and provides the flexibility of modeling different time series structures with different inputs. Given that the interactions between $\mathbf{x}$ and time are captured by $\textbf{x}y_{t-l}$, $Z_t$ is assumed to vary independently over time to reduce modeling and computational complexity. Further extensions can be made by replacing $Z_t(\mathbf{x})$ with a spatio-temporal Gaussian process $Z(t,\mathbf{x})$, but the computational cost will be higher.  
Without the Gaussian process assumption in (\ref{eq:model_GPBiTS}), the mean function is closely related to the Zeger-Qaqish (\citeyear{zeger1988markov}) model and its extensions in \cite{hung2012binary} and \cite{benjamin2003generalized}, all of which take into account the autoregressive predictors in logistic regression.

%This model includes several important models as special cases. When there is only one binary response for each input setting (i.e., $T=1$), model (\ref{eq:model_GPBiTS}) can be written as
%\begin{equation}\label{eq:simplemodel}
%\text{logit}(p(\mathbf{x}))=\mu(\mathbf{x})+Z(\mathbf{x}),
%\end{equation}
%where the mean function is $\mu(\mathbf{x})=\alpha_0+\mathbf{x}'\boldsymbol{\alpha}$. This is a generalization of the GP model to binary data with logistic link. A special case of (\ref{eq:simplemodel}) where $\mathbf{x}$ is limited to the two-dimensional spatial domain is called the spatial generalized linear mixed model \citep{zhang2002estimation}. Although this type of generalization of GP has been applied in many different field (references), theoretical studies on the estimation and prediction properties is not yet available.  

%There are several special cases of (\ref{eq:model_GPBiTS}) discussed in the literature which focus on modeling the serial correlations in binary outputs. For example, Zeger-Qaqish (\citeyear{zeger1988markov}) model takes into account the autoregressive predictors in logistic regression,
%$$
%\text{logit}(p_{t}(\mathbf{x}))=\eta_{t}=\sum^R_{r=1}\varphi_ry_{t-r}+\alpha_0+\mathbf{x}'\boldsymbol{\alpha},
%$$
%and
%\cite{hung2012binary} incorporates autoregressive predictors in a generalized linear mixed model.
Model (\ref{eq:model_GPBiTS}) extends the applications of conventional GP to binary time series generated from computer experiments. The model is intuitively appealing; however, the issues of estimation, prediction, and inference are not straightforward due to the nature of binary response and the dependency structure. %Therefore, this research focuses on  and the theoretical properties of the estimated parameters  systematically studied. Furthermore, the prediction and the construction of predictive distribution is crucial for the analysis of computer experiments but it is not available in the literature.  

\section{Inference}
Since model \eqref{eq:simplemodel} can be written as a special case of model \eqref{eq:model_GPBiTS} when $R=0,L=0$ and $T=1$, derivations herein are mainly based on model \eqref{eq:model_GPBiTS} with additional discussions given for \eqref{eq:simplemodel} when necessary.  

\subsection{Estimation}\label{sec:estimation}
Given $n$ input settings $\mathbf{x}_1,\ldots,\mathbf{x}_n$ in a computer experiment, denote $y_{it}\equiv y_t(\mathbf{x}_i)$ as the binary output generated from input $\mathbf{x}_i$ at time $t$, where $\mathbf{x}_i\in \mathbb{R}^d$, $i=1,...,n$, and $t=1,...,T$. Let $N$ be the total number of the outputs, i.e., $N=nT$. In addition, at each time $t$, denote $\mathbf{y}_t$ as an $n$-dimensional vector $\mathbf{y}_t=(y_{1t},...,y_{nt})'$ with conditional mean $\boldsymbol{p}_{t}=(p_{1t},\ldots,p_{nt})'$, where $p_{it}=\mathbb{E}(y_{it}|H_{it})$ and $H_{it}=\{y_{i,t-1},y_{i,t-2},\ldots\}$.  Based on the data, model \eqref{eq:model_GPBiTS} can be rewritten into matrix form as follows:
\begin{equation}\label{eq:modelwithdata}
\text{logit}(\boldsymbol{p})=\boldsymbol{X}\boldsymbol{\beta}+\boldsymbol{Z},\quad \boldsymbol{Z}\sim\mathcal{N}(\mathbf{0}_N,\Sigma(\boldsymbol{\omega})),
\end{equation}
where $\boldsymbol{p}=(\boldsymbol{p}'_1,\ldots,\boldsymbol{p}'_T)',\boldsymbol{\beta}=(\varphi_1,\ldots,\varphi_R,\alpha_0,\boldsymbol{\alpha}',(\boldsymbol{\gamma}_1',\ldots,\boldsymbol{\gamma}_L')')',\boldsymbol{\omega}=(\sigma^2,\boldsymbol{\theta})',\\
\boldsymbol{Z}=(Z_1(\mathbf{x}_1),\ldots,Z_1(\mathbf{x}_n),\ldots,Z_T(\mathbf{x}_1),\ldots,Z_T(\mathbf{x}_n))'$,  $\boldsymbol{X}$ is the model matrix $(X'_1,\ldots,X'_T)'$, $X_t$ is an $n\times(1+R+d+dL)$ matrix with $i$-th row defined by
$(X_t)_i=(1,y_{i,t-1},\ldots,y_{i,t-R},\mathbf{x}_i',\\\mathbf{x}'_iy_{i,t-1},\ldots,\mathbf{x}'_iy_{i,t-L})$, and 
$\Sigma(\boldsymbol{\omega})$ is an $N\times N$ covariance matrix defined by
\begin{equation}\label{eq:Sigma}
\Sigma(\boldsymbol{\omega})=\sigma^2\boldsymbol{R}_{\boldsymbol{\theta}}\otimes I_T
\end{equation}
with $(\boldsymbol{R}_{\boldsymbol{\theta}})_{ij}=R_{\boldsymbol{\theta}}(\mathbf{x}_i,\mathbf{x}_j)$. Model \eqref{eq:simplemodel} can also be rewritten in the same way by setting $R=0,L=0$ and $T=1$.
 
With the presence of time series and their interaction with the input settings in model (\ref{eq:model_GPBiTS}), we can write down the partial likelihood (PL) function \citep{cox1972partial,cox1975partial} according to the formulation of \cite{slud1994partial}. Given the previous information $\{H_{it}\}_{i=1,\ldots,n;t=1,\ldots,N}$, the PL for  $\boldsymbol{\beta}$ can be written as
\begin{equation}\label{eq:fixedeffectpartiallikelihood}
PL(\boldsymbol{\beta}|\boldsymbol{Z})=\prod^n_{i=1}\prod^T_{t=1}(p_{it}(\boldsymbol{\beta}|\boldsymbol{Z}))^{y_{it}}(1-p_{it}(\boldsymbol{\beta}|\boldsymbol{Z}))^{1-y_{it}},
\end{equation}
where $p_{it}(\boldsymbol{\beta}|\boldsymbol{Z})=\mathbb{E}_{\boldsymbol{\beta}|\boldsymbol{Z}}[y_{it}|H_{it}]$. Then, the integrated quasi-PL function for the estimation of $(\boldsymbol{\beta},\boldsymbol{\omega})$ is given by
\begin{equation}\label{eq:fixedeffectquasipartiallikelihood}
|\Sigma(\boldsymbol{\omega})|^{-1/2}\int\exp\{\log PL(\boldsymbol{\beta}|\boldsymbol{Z})-\frac{1}{2}\boldsymbol{Z}'\Sigma(\boldsymbol{\omega})^{-1}\boldsymbol{Z} \}d\boldsymbol{Z}.
\end{equation}
Note that, for model \eqref{eq:simplemodel} where no time series effect is considered, \eqref{eq:fixedeffectpartiallikelihood} and \eqref{eq:fixedeffectquasipartiallikelihood} should be replaced by the likelihood function 
\begin{equation*}
L(\boldsymbol{\beta}|\boldsymbol{Z})=\prod^n_{i=1}(p_{i1}(\boldsymbol{\beta}|\boldsymbol{Z}))^{y_{i1}}(1-p_{i1}(\boldsymbol{\beta}|\boldsymbol{Z}))^{1-y_{i1}}
\end{equation*}
and the integrated quasi-likelihood function
\begin{equation}\label{eq:fixedeffectquasilikelihood}
|\Sigma(\boldsymbol{\omega})|^{-1/2}\int\exp\{\log L(\boldsymbol{\beta}|\boldsymbol{Z})-\frac{1}{2}\boldsymbol{Z}'\Sigma(\boldsymbol{\omega})^{-1}\boldsymbol{Z} \}d\boldsymbol{Z},
\end{equation}
respectively. Hereafter, we provide the framework for the integrated quasi-PL function \eqref{eq:fixedeffectquasipartiallikelihood}, but the result can be applied to the integrated quasi-likelihood function \eqref{eq:fixedeffectquasilikelihood} by assuming $R=0,L=0$ and $T=1$.

Because of the difficulty in computing the integrated quasi-PL function, a \textit{penalized quasi-PL} (PQPL) function is used as an approximation. Similar to the procedure in \cite{breslow1993approximate}, the integrated quasi-partial log-likelihood can be approximated by Laplace's method \citep{barndorff1989asymptotic}. Ignoring the multiplicative constant and plugging \eqref{eq:Sigma} in $\Sigma(\boldsymbol{\omega})$, the approximation yields
\begin{align}\label{eq:PQPL}
\nonumber
-\frac{1}{2}\log|I_n&+\sigma^2\boldsymbol{W}(\boldsymbol{R}_{\boldsymbol{\theta}}\otimes I_T)|+\\
&\sum^n_{i=1}\sum^T_{t=1}\left(y_{it}\log\frac{p_{it}(\boldsymbol{\beta}|\tilde{\boldsymbol{Z}})}{1-p_{it}(\boldsymbol{\beta}|\tilde{\boldsymbol{Z}})}+\log(1-p_{it}(\boldsymbol{\beta}|\tilde{\boldsymbol{Z}}))\right)-\frac{1}{2\sigma^2}\tilde{\boldsymbol{Z}}'(\boldsymbol{R}_{\boldsymbol{\theta}}\otimes I_T)^{-1}\tilde{\boldsymbol{Z}},
\end{align}
where $\boldsymbol{W}$ is an $N\times N$ diagonal matrix with diagonal elements $W_{it}=p_{it}(\boldsymbol{\beta}|\tilde{\boldsymbol{Z}})(1-p_{it}(\boldsymbol{\beta}|\tilde{\boldsymbol{Z}}))$, $p_{it}(\boldsymbol{\beta}|\tilde{\boldsymbol{Z}})=\mathbb{E}_{\boldsymbol{\beta}|\tilde{\boldsymbol{Z}}}[y_{it}|H_{it}]$, and $\tilde{\boldsymbol{Z}}=\tilde{\boldsymbol{Z}}(\boldsymbol{\beta},\boldsymbol{\omega})$ is the solution of
$
\sum^n_{i=1}\sum^T_{t=1}\mathbf{e}_{it}(y_{it}-p_{it}(\boldsymbol{\beta}|\boldsymbol{Z}))=(\boldsymbol{R}_{\boldsymbol{\theta}}\otimes I_T)^{-1}\boldsymbol{Z}/\sigma^2
$, where $\mathbf{e}_{it}$ is a unit-vector where $((t-1)n+i)$-th element is one. The estimator $\hat{\boldsymbol{\beta}}$ which maximizes the PQPL function \eqref{eq:PQPL} is called \textit{maximum quasi-PL estimator}. Thus, similar to the derivations in \cite{breslow1993approximate} for score equations of a penalized quasi-likelihood function, the score equations of the PQPL function for $\boldsymbol{\beta}$ and $\boldsymbol{\omega}$ are 
$$
\sum^n_{i=1}\sum^T_{t=1}X_{it}(y_{it}-p_{it}(\boldsymbol{\beta},\boldsymbol{\omega}))=0
$$
and 
$$
\sum^n_{i=1}\sum^T_{t=1}\mathbf{e}_{it}(y_{it}-p_{it}(\boldsymbol{\beta},\boldsymbol{\omega}))=(\boldsymbol{R}_{\boldsymbol{\theta}}\otimes I_T)^{-1}\boldsymbol{Z}/\sigma^2,
$$
where $p_{it}(\boldsymbol{\beta},\boldsymbol{\omega})=\mathbb{E}_{\boldsymbol{\beta},\boldsymbol{\omega}}[y_{it}|H_{it}]$. The solution to the score equations can be efficiently obtained by an iterated weighted least squares (IWLS) approach as follows. In each step, one first solves for $\boldsymbol{\beta}$ in 
\begin{equation}\label{eq:estimation1}
(\boldsymbol{X}'\boldsymbol{V}(\boldsymbol{\omega})^{-1}\boldsymbol{X})\boldsymbol{\beta}=\boldsymbol{X}'\boldsymbol{V}(\boldsymbol{\omega})^{-1}\tilde{\boldsymbol{\eta}},
\end{equation}
where $\boldsymbol{V}(\boldsymbol{\omega})=\boldsymbol{W}^{-1}+\sigma^2(\boldsymbol{R}_{\boldsymbol{\theta}}\otimes I_T)$, $\boldsymbol{W}$ is an $N\times N$ diagonal matrix with diagonal elements $W_{it}=p_{it}(\boldsymbol{\beta},\boldsymbol{\omega})(1-p_{it}(\boldsymbol{\beta},\boldsymbol{\omega}))$, and $\tilde{\eta}_{it}=\log\frac{p_{it}(\boldsymbol{\beta},\boldsymbol{\omega})}{1-p_{it}(\boldsymbol{\beta},\boldsymbol{\omega})}+\frac{y_{it}-p_{it}(\boldsymbol{\beta},\boldsymbol{\omega})}{p_{it}(\boldsymbol{\beta},\boldsymbol{\omega})(1-p_{it}(\boldsymbol{\beta},\boldsymbol{\omega}))}$, and then sets 
\begin{equation}\label{eq:estimation2}
\hat{\boldsymbol{Z}}=\sigma^2(\boldsymbol{R}_{\boldsymbol{\theta}}\otimes I_T)\boldsymbol{V}(\boldsymbol{\omega})^{-1}(\tilde{\boldsymbol{\eta}}-\boldsymbol{X}'\hat{\boldsymbol{\beta}})
\end{equation}
and replaces $p_{it}(\boldsymbol{\beta},\boldsymbol{\omega})$ with $p_{it}(\hat{\boldsymbol{\beta}},\boldsymbol{\omega})=\left(\frac{\exp\{\boldsymbol{X}'\hat{\boldsymbol{\beta}}+\hat{\boldsymbol{Z}}\}}{\mathbf{1}_N+\exp\{\boldsymbol{X}'\hat{\boldsymbol{\beta}}+\hat{\boldsymbol{Z}}\}}\right)_{it}$.

Estimation of the correlation parameters $\boldsymbol{\theta}$ and variance $\sigma^2$ is obtained by the restricted maximum likelihood (REML) approach \citep{patterson1971recovery} because it is known to have smaller bias comparing with the maximum likelihood approach \citep{patterson1974maximum}. See also \cite{harville1977maximum} and \cite{searle2009variance} for details. According to \cite{harville1974bayesian,harville1977maximum}, the REML estimators of $\sigma^2$ and $\boldsymbol{\theta}$ can be solved by minimizing the following negative log-likelihood function with respect to $\boldsymbol{\omega}$,
\begin{equation}\label{eq:estimation3}
L(\boldsymbol{\omega})=\frac{N-m}{2}\log(2\pi)-\frac{1}{2}\log(|\boldsymbol{X}'\boldsymbol{X}|)+\frac{1}{2}\log(|\boldsymbol{V}(\boldsymbol{\omega})|)+\frac{1}{2}\log(|\boldsymbol{X}'\boldsymbol{V}(\boldsymbol{\omega})^{-1}\boldsymbol{X}|)+\frac{1}{2}\tilde{\boldsymbol{\eta}}'\varPi(\boldsymbol{\omega})\tilde{\boldsymbol{\eta}},
\end{equation}
where $m=1+R+d+dL$ and  $\varPi(\boldsymbol{\omega})=\boldsymbol{V}(\boldsymbol{\omega})^{-1}-\boldsymbol{V}(\boldsymbol{\omega})^{-1}\boldsymbol{X}(\boldsymbol{X}'\boldsymbol{V}(\boldsymbol{\omega})^{-1}\boldsymbol{X})^{-1}\boldsymbol{X}'\boldsymbol{V}(\boldsymbol{\omega})^{-1}$. 

Therefore, the estimators $\hat{\boldsymbol{\beta}}$ and $\hat{\boldsymbol{\omega}}$ $(\equiv(\hat{\sigma}^2,\hat{\boldsymbol{\theta}})')$ can be obtained by iteratively solving \eqref{eq:estimation1}, \eqref{eq:estimation2} and minimizing \eqref{eq:estimation3}. The explicit algorithm is given in the supplementary material S1. Note that $\boldsymbol{V}(\boldsymbol{\omega})$ is a block diagonal matrix, i.e., a square matrix having main diagonal blocks square matrices such that the off-diagonal blocks are zero matrices. Therefore the computational burden for the matrix inversion of $\boldsymbol{V}(\boldsymbol{\omega})$ can be alleviated by the fact that the inverse of a block diagonal matrix is a block diagonal matrix, composed of the inversion of each block.

\iffalse
\[
\text{tr}\{\boldsymbol{\Gamma}(\partial\boldsymbol{V}/\partial\sigma^2)\}+\tilde{\boldsymbol{\eta}}'(\partial \boldsymbol{\Gamma}/\partial\sigma^2)\tilde{\boldsymbol{\eta}}=0
\]
and 
\[
\text{tr}\{\boldsymbol{\Gamma}(\partial\boldsymbol{V}/\partial\theta_i)\}+\tilde{\boldsymbol{\eta}}'(\partial \boldsymbol{\Gamma}/\partial\theta_i)\tilde{\boldsymbol{\eta}}=0,\quad i=1,\ldots,d,
\]
where $\boldsymbol{\Gamma}=\boldsymbol{V}^{-1}-\boldsymbol{V}^{-1}\boldsymbol{X}(\boldsymbol{X}'\Sigma^{-1}\boldsymbol{X})^{-1}\boldsymbol{X}'\boldsymbol{V}^{-1}$.
\fi

\subsection{Asymptotic Properties}

Asymptotic results are presented here to show that the estimators $\hat{\boldsymbol{\beta}},\hat{\sigma}^2$ and $\hat{\boldsymbol{\theta}}$ obtained in Section \ref{sec:estimation} are asymptotically normally distributed when $N(=nT)$ becomes sufficiently large. In the present context both $n$ and $T$ are sufficiently large. The assumptions are given in the supplementary material S2, and the proofs are stated in the supplementary material S3 and S4. These results are developed along the lines described in \cite{hung2012binary} and \cite{cressie1993asymptotic,cressie1996asymptotics}.

\begin{theorem}\label{thm:asy_fixed}
Under assumptions S2.1 and S2.2, the maximum quasi-PL estimator for the fixed effects $\boldsymbol{\beta}$ are consistent and asymptotically normal as $N\rightarrow\infty$,
\[
\sqrt{N}(\hat{\boldsymbol{\beta}}-\boldsymbol{\beta})=\Lambda_N^{-1}\frac{1}{\sqrt{N}}S_N(\boldsymbol{\beta},\boldsymbol{\omega})+o_p(1)
\]
and
\[
\sqrt{N}\Lambda_N^{1/2}(\hat{\boldsymbol{\beta}}-\boldsymbol{\beta})\stackrel{d}{\longrightarrow}\mathcal{N}(\mathbf{0},I_m),
\]
where $m$ is the size of the vector $\boldsymbol{\beta}$ (i.e., $m=1+R+d+dL$), the sample information matrix $$\Lambda_N=\frac{1}{N}\sum^n_{i=1}\sum^T_{t=1}X_{it}X'_{it}p_{it}(\boldsymbol{\beta},\boldsymbol{\omega})(1-p_{it}(\boldsymbol{\beta},\boldsymbol{\omega})),$$ and $S_N(\boldsymbol{\beta},\boldsymbol{\omega})=\sum^n_{i=1}\sum^T_{t=1}X_{it}(y_{it}-p_{it}(\boldsymbol{\beta},\boldsymbol{\omega}))$.
\end{theorem} 

\begin{remark} For model \eqref{eq:simplemodel}, the estimator $\hat{\boldsymbol{\beta}}$ can be obtained by minimizing the penalized quasi-likelihood (PQL) function, which can be written as \eqref{eq:PQPL} with $T=1$. Under assumption S2.1 and the application of central limit theorem, such estimator has the same asymptotic properties as in Theorem \ref{thm:asy_fixed} with $N=n$.  
\end{remark}

For models ({\ref{eq:simplemodel}) and (\ref{eq:model_GPBiTS}), we have the following asymptotic properties for
$\hat{\boldsymbol{\omega}}$.

\begin{theorem}\label{thm:asy_reml}
Denote $[\Gamma_N(\boldsymbol{\omega})]_{i,j}=\partial^2L(\boldsymbol{\omega})/\partial\omega_i\partial\omega_j$ and $J_N(\boldsymbol{\omega})=[\mathbb{E}_{\boldsymbol{\omega}}\Gamma_N(\boldsymbol{\omega})]^{1/2}$. Then, under assumptions S2.3 and S2.4, as $N\rightarrow\infty$,
\[
J_N(\hat{\boldsymbol{\omega}})(\hat{\boldsymbol{\omega}}-\boldsymbol{\omega})\stackrel{d}{\longrightarrow}\mathcal{N}(\mathbf{0},I_{d+1}).
\]
\end{theorem}

%\section{GPBiTS Model and Best MSPE Predictor}
%\subsection{GPBiTS Model}
%\begin{definition}\label{def:model}
%Denote $p_{t}=E(y_t|H_t)$, $H_t=\{y_{t-1},y_{t-2},\ldots,p_{t-1},p_{t-2},\ldots\}$. Consider a model in the following:
%\begin{equation}\label{eq:model_GPBiTS}
%\text{logit}(p_{t}(\mathbf{x}))=\eta_{t}=\sum^R_{r=1}\varphi_ry_{t-r}+\sum^Q_{q=1}\zeta_q(y_{t-q}-p_{t-q})+\mathbf{x}'\boldsymbol{\alpha}+Z(\mathbf{x})+\sum^L_{l=1}\gamma_l\mathbf{x}y_{t-l},
%\end{equation}
%where $Z(\cdot)\sim GP(0,\sigma^2R_\theta(\cdot,\cdot))$ and $R_\theta$ is the correlation function with parameters $\theta$.
%This model is called \textit{Gaussian process model with binary time series} (GPBiTS) in the context.
%\end{definition}

Note that the asymptotic results here focus on the conditional inference of $\hat{\boldsymbol{\beta}}|\boldsymbol{\omega}$ and $\hat{\boldsymbol{\omega}}|\boldsymbol{\beta}$. Therefore, these results still hold in the presence of the orthogonal Gaussian process approach proposed by \cite{plumlee2016orthogonal}. Theoretically speaking, the orthogonal Gaussian process approach is expected to reduce the covariance between $\hat{\boldsymbol{\beta}}$ and $\hat{\boldsymbol{\omega}}$. However, the derivation of the joint distribution would be nontrivial. We leave this result to the future work.

\section{Construction of Predictive Distribution}

For computer experiments, the construction of an optimal predictor and its corresponding predictive distribution is important for uncertainty quantification, sensitivity analysis, process optimization, and calibration \citep{santner2003design}.

First, some notation is introduced. 
For some untried setting $\mathbf{x}_{n+1}$, denote the predictive probability at time $s$ by $p_s(\mathbf{x}_{n+1})=\mathbb{E}[y_s(\mathbf{x}_{n+1})|H_{s}]$, where $H_{s}=\{y_{n+1,s-1},y_{n+1,s-2},\ldots\}$. 
Assume that $D_{n+1,s}$ represents the ``previous information" including $\{y_{n+1,s-1},y_{n+1,s-2},\ldots,\\p_{n+1,s-1},p_{n+1,s-2},\ldots \}$ at $\mathbf{x}_{n+1}$ and $\{y_{it},p_{it}\}$, where $i=1,\ldots,n$ and $t=1,\ldots,T$.
Also, let $Logitnormal(\mu, \sigma^2)$ represent a logit-normal distribution $P$, where $P=\exp\{X\}/(1+\exp\{X\})$ and $X$ has a univariate normal distribution with $\mu$ and variance $\sigma^2$.
Denote the first two moments of the distribution by $\mathbb{E}[P]=\kappa(\mu,\sigma^2)$ and $\mathbb{V}[P]=\tau(\mu,\sigma^2)$. In general, there is no closed form expression for $\kappa(\mu,\sigma^2)$ and $\tau(\mu,\sigma^2)$, but it can be easily computed by numerical integration such as in the package \texttt{logitnorm} \citep{logitnorm} in \texttt{R} \citep{R}.  More discussions on logit-normal distribution can be found in \cite{mead1965generalised,atchison1980logistic,frederic2008two}.

We first present a lemma which shows that, given $D_{n+1,s}$, the conditional distribution of  $p_s(\mathbf{x}_{n+1})$ in model \eqref{eq:model_GPBiTS} is logit-normal.   
This result lays the foundation for the construction of predictive distribution. The proof is given in the supplementary material S5. %The result can be applied to model \eqref{eq:simplemodel}, a special case of \eqref{eq:model_GPBiTS} without time-series, in Remark \ref{lemma}. 

\begin{lemma}\label{thm1} For model \eqref{eq:model_GPBiTS}, the conditional distribution of $p_s(\mathbf{x}_{n+1})$ can be written as
\[
p_s(\mathbf{x}_{n+1})|D_{n+1,s}\sim Logitnormal(m(D_{n+1,s}),v(D_{n+1,s})),
\]
where 
\begin{align*}
m(D_{n+1,s})=\sum^R_{r=1}\varphi_ry_{n+1,s-r}+\alpha_0+\mathbf{x}'_{n+1}\boldsymbol{\alpha}+\sum^L_{l=1}\boldsymbol{\gamma}_l\mathbf{x}_{n+1}y_{n+1,s-l}+\boldsymbol{r}_{\boldsymbol{\theta}}' \boldsymbol{R}_{\boldsymbol{\theta}}^{-1}\left(\log \frac{\boldsymbol{p}_s}{\boldsymbol{1}_n-\boldsymbol{p}_s}-\boldsymbol{\mu}_s\right),
\end{align*}
$v(D_{n+1,s})=\sigma^2\left(1-\boldsymbol{r}_{\boldsymbol{\theta}}' \boldsymbol{R}_{\boldsymbol{\theta}}^{-1}\boldsymbol{r}_{\boldsymbol{\theta}}\right),$
$\boldsymbol{r}_{\boldsymbol{\theta}}=(R_{\boldsymbol{\theta}}(\mathbf{x}_{n+1},\mathbf{x}_1),\ldots,R_{\boldsymbol{\theta}}(\mathbf{x}_{n+1},\mathbf{x}_n))',\boldsymbol{R}_{\boldsymbol{\theta}}=\{R_{\boldsymbol{\theta}}(\mathbf{x}_i,\mathbf{x}_j)\}$, $\boldsymbol{p}_s=(p_s(\mathbf{x}_1),\ldots,p_s(\mathbf{x}_n))'$,  and 
$
(\boldsymbol{\mu}_s)_i=\sum^R_{r=1}\varphi_ry_{i,s-r}+\alpha_0+\mathbf{x}_i'\boldsymbol{\alpha}+\sum^L_{l=1}\boldsymbol{\gamma}_l\mathbf{x}_iy_{i,s-l}.
$
\end{lemma}

\begin{remark}\label{lemma}
For model \eqref{eq:simplemodel}, the result in Lemma \ref{thm1} can be applied by having $R=0,L=0,s=1$ and $T=1$. Then, $D_{n+1,s}$ can be written as $D_{n+1}$ containing only $\{p_{1,1},\ldots,p_{n,1}\}$, and we have the conditional distribution
$$
p(\mathbf{x}_{n+1})|D_{n+1}\sim Logitnormal(m(D_{n+1}), v(D_{n+1})),
$$
where
$m(D_{n+1})=\alpha_0+\mathbf{x}'_{n+1}\boldsymbol{\alpha}+\boldsymbol{r}_{\boldsymbol{\theta}}'\boldsymbol{R}_{\boldsymbol{\theta}}^{-1}(\log \frac{\boldsymbol{p}_1}{\mathbf{1}-\boldsymbol{p}_1}-\boldsymbol{\mu}^n)$,
$
v(D_{n+1})=\sigma^2(1-\boldsymbol{r}_{\boldsymbol{\theta}}'\boldsymbol{R}_{\boldsymbol{\theta}}^{-1}\boldsymbol{r}_{\boldsymbol{\theta}}),$ $\boldsymbol{\mu}^n=(\alpha_0+\mathbf{x}'_1\boldsymbol{\alpha},\ldots,\alpha_0+\mathbf{x}'_n\boldsymbol{\alpha})',\boldsymbol{r}_{\boldsymbol{\theta}}=(R_{\boldsymbol{\theta}}(\mathbf{x}_{n+1},\mathbf{x}_1),\ldots,R_{\boldsymbol{\theta}}(\mathbf{x}_{n+1},\mathbf{x}_n))'$, and $\boldsymbol{R}_{\boldsymbol{\theta}}=\{R_{\boldsymbol{\theta}}(\mathbf{x}_i,\mathbf{x}_j)\}$.
\end{remark}

Based on Lemma \ref{thm1}, the prediction of $p_s(\mathbf{x}_{n+1})$ for some untried setting $\mathbf{x}_{n+1}$ and its variance can then be obtained in the next theorem. The proof is given in the supplementary material S6. The definition of minimum mean squared prediction error of $p$ given $D$ is first stated as follows,
\[
\hat{p}=\hat{p}(D)=\arg\min_{\eta}\mathbb{E}_F[(p-\eta)^2],
\]
where $F(\cdot)$ is the joint distribution of $(p,D)$ and $\mathbb{E}_F[\cdot]$ denotes expectation under distribution $F(\cdot)$.

\begin{theorem}\label{thm2}
Given $D_{n+1,s}=\{\boldsymbol{y}'_1,\ldots,\boldsymbol{y}'_T, \boldsymbol{p}'_1,\ldots,\boldsymbol{p}'_T,y_{n+1,s-1},\ldots,y_{n+1,1},p_{n+1,s-1},\ldots,p_{n+1,1}\}$,
\begin{itemize} 
\item[(i)] the minimum mean squared prediction error (MMSPE) predictor of $p_s(\mathbf{x}_{n+1})$, denoted by $\hat{p}_s(\mathbf{x}_{n+1})$, is 
\begin{align*}\label{MMSE}
\mathbb{E}\left[p_s(\mathbf{x}_{n+1})|D_{n+1,s}\right]&=\kappa(m(D_{n+1,s}),v(D_{n+1,s}))\\
\text{with variance}\quad 
\mathbb{V}\left[p_s(\mathbf{x}_{n+1})|D_{n+1,s}\right]&=\tau(m(D_{n+1,s}),v(D_{n+1,s}));
\end{align*}
\item[(ii)] the MMSPE predictor is an interpolator, i.e., if $\mathbf{x}_{n+1}=\mathbf{x}_{i}$ for $i=1,\cdots, n$, then $\hat{p}_s(\mathbf{x}_{n+1})=\mathbb{E}\left[p_s(\mathbf{x}_{n+1})|D_{n+1,s}\right]=p_s(\mathbf{x}_{i})$ and the predictive variance is 0;
\item[(iii)] the $q$-th quantile of the conditional distribution $p(\mathbf{x}_{n+1})|D_{n+1,s}$ is
\[
\frac{\exp\{m(D_{n+1,s})+z_q\sqrt{v(D_{n+1,s})}\}}{1+\exp\{m(D_{n+1,s})+z_q\sqrt{v(D_{n+1,s})}\}},
\]
where $z_q$ is the $q$-th quantile of the standard normal distribution. 
\end{itemize}
\end{theorem}

Theorem \ref{thm2} shows that, given $D_{n+1,s}$, the new predictor for binary data can interpolate the underlying probabilities which generate the training data. According to Theorem \ref{thm2}(iii) and the fact that $v(D_{n+1,s})$ increases with the distance to the training data, 
this result shows an increasing predictive uncertainty for points away from the training data. This predictive property is desirable and consistent with the conventional GP predictor.

In practice, only the binary outputs are observable and the underlying probabilities are not available in the training data. Thus, the following results construct the MMSPE predictor of $p_s(\mathbf{x}_{n+1})$ given $\boldsymbol{Y}=(\boldsymbol{y}'_1,\ldots,\boldsymbol{y}'_T, y_1(\mathbf{x}_{n+1}),\ldots,y_{s-1}(\mathbf{x}_{n+1}))'$. These results can be used for prediction and quantification of the predictive uncertainty, such as constructing predictive confidence intervals for untried settings.

\begin{theorem}\label{cor1}
Given $\boldsymbol{Y}=(\boldsymbol{y}'_1,\ldots,\boldsymbol{y}'_T, y_1(\mathbf{x}_{n+1}),\ldots,y_{s-1}(\mathbf{x}_{n+1}))'$,
\begin{itemize} 
\item[(i)]
The MMSPE predictor of $p_s(\mathbf{x}_{n+1})$ is 
\begin{equation}\label{cor:mean}
\hat{p}_s(\mathbf{x}_{n+1})=\mathbb{E}\left[p_s(\mathbf{x}_{n+1})|\boldsymbol{Y}\right]=\mathbb{E}_{\boldsymbol{p}|\boldsymbol{Y}}\left[\kappa(m(D_{n+1,s}),v(D_{n+1,s}))|\boldsymbol{Y}\right]
\end{equation}
with variance
\begin{equation}\label{cor:var}
\mathbb{V}\left[p_s(\mathbf{x}_{n+1})|\boldsymbol{Y}\right]=\mathbb{E}_{\boldsymbol{p}|\boldsymbol{Y}}\left[\tau(m(D_{n+1,s}),v(D_{n+1,s}))|\boldsymbol{Y}\right]+\mathbb{V}_{\boldsymbol{p}|\boldsymbol{Y}}\left[\kappa(m(D_{n+1,s}),v(D_{n+1,s}))|\boldsymbol{Y}\right],
\end{equation}
where $\boldsymbol{p}=(\boldsymbol{p}'_1,\cdots,\boldsymbol{p}'_T, p_1(\mathbf{x}_{n+1}),\ldots,p_{s-1}(\mathbf{x}_{n+1}))'$.
\item[(ii)]  When $\mathbf{x}_{n+1}=\mathbf{x}_i$, the MMSPE predictor becomes  $\hat{p}_s(\mathbf{x}_{i})=\mathbb{E}_{\boldsymbol{p}|\boldsymbol{Y}}\left[p_s(\mathbf{x}_i)|\boldsymbol{Y}\right]$  with variance $\mathbb{V}_{\boldsymbol{p}|\boldsymbol{Y}}\left[p_s(\mathbf{x}_i)|\boldsymbol{Y}\right]$. %which is due to the uncertainty of the underlying probability. 
\item[(iii)]  The quantiles of $p_s(\mathbf{x}_{n+1})|\boldsymbol{Y}$ can be approximated by the sample quantiles of $\{p^{(j)}_s\}$, where  $p^{(j)}_s$ are random samples generated from $p_s(\mathbf{x}_{n+1})|D_{n+1,s}$ following the logit-normal distribution given in Lemma \ref{thm1}, and $\boldsymbol{p}$ in $D_{n+1,s}=\{\boldsymbol{p},\boldsymbol{Y}\}$ are random samples from distribution $\boldsymbol{p}|\boldsymbol{Y}$. 
\end{itemize}
\end{theorem}

Although the outcomes in the study are binary, the major focus is on the predictive uncertainty of the underlying probability $p_s(\mathbf{x}_{n+1})$. There are two reasons. First, it is more informative and usually of scientific interest to predict generating probabilities instead of binary outcomes. Second, given the understanding of predictive distribution for $p_s(\mathbf{x}_{n+1})$ in Theorem \ref{cor1}, the predictive uncertainty of $y_s(\mathbf{x}_{n+1})$ can be easily constructed by the use of bootstrap predictive distribution, which is given in the following remark. 

\begin{remark}
 The bootstrap predictive distribution of $y_s(\mathbf{x}_{n+1})$ follows $Bernoulli(p_s(\mathbf{x}_{n+1}))$, where the distribution of $p_s(\mathbf{x}_{n+1})$ can be obtained from Theorem \ref{cor1} (iii). The detailed algorithm can be found in the supplementary material S7.  
\end{remark}

\begin{remark}
For model \eqref{eq:simplemodel}, the results of Theorem \ref{thm2} and Theorem \ref{cor1} can be applied by assuming $s=1$ and $T=1$.  

%$(y_1,\ldots,y_n)'=\boldsymbol{y}^n$, the minimum mean squared error (MMSE) predictor of $p(\mathbf{x}_{n+1})$ is 
%$$\mathbb{E}\left[p(\mathbf{x}_{n+1})|\boldsymbol{y}^n\right]=\mathbb{E}_{\boldsymbol{p}|\boldsymbol{y}^n}\left[\kappa(m(D_{n+1,s}),v(D_{n+1,s}))|\boldsymbol{y}^n\right]$$
%with variance
%$$\mathbb{V}\left[p(\mathbf{x}_{n+1})|\boldsymbol{y}^n\right]=\mathbb{E}_{\boldsymbol{p}^n|\boldsymbol{y}^n}\left[\tau(m(\boldsymbol{p}^n),v(\boldsymbol{p}^n))|\boldsymbol{y}^n\right]+\mathbb{V}_{\boldsymbol{p}^n|\boldsymbol{y}^n}\left[\kappa(m(\boldsymbol{p}^n),v(\boldsymbol{p}^n))|\boldsymbol{y}^n\right].$$
\end{remark}

Although there is no closed form expression for the distribution of $\boldsymbol{p}|\boldsymbol{Y}$, the random samples from $\boldsymbol{p}|\boldsymbol{Y}$ can be easily generated by the Metropolis-Hastings (MH) algorithm.   For example, let $\{\boldsymbol{p}^{(j)}\}_{j=1,\ldots,J}$ be the $J$ random samples generated from distribution $\boldsymbol{p}|\boldsymbol{Y}$, then the MMSPE predictor of $p_s(\mathbf{x}_{n+1})$ in Theorem \ref{cor1} can be approximated using Monte Carlo method by
\[
\mathbb{E}_{\boldsymbol{p}|\boldsymbol{Y}}\left[\kappa(m(D_{n+1,s}),v(D_{n+1,s}))|\boldsymbol{Y}\right]\approx\frac{1}{J}\sum^J_{j=1}\kappa(m(D^{(j)}_{n+1,s}),v(D^{(j)}_{n+1,s})),
\]
where $D^{(j)}_{n+1,s}=\{\boldsymbol{p}^{(j)},\boldsymbol{Y}\}$. Similar idea can be applied to compute $\mathbb{V}\left[p_s(\mathbf{x}_{n+1})|\boldsymbol{Y}\right]$ and the predictive quantiles. Details are given in the supplementary material S7.

Without the information of the underlying probabilities, the predictor does not interpolate all the training data as in Theorem \ref{thm2} (ii). From Theorem \ref{cor1}, when $\mathbf{x}_{n+1}=\mathbf{x}_i$, the predictor is still unbiased but the corresponding variance is nonzero. Instead, the variance becomes $\mathbb{V}_{\boldsymbol{p}|\boldsymbol{Y}}\left[p_s(\mathbf{x}_i)|\boldsymbol{Y}\right]$, which is due to the uncertainty of the underlying probability. The proof is similar to Theorem \ref{thm2} (ii). 
To show the empirical performance of the predictive distribution in Theorem \ref{cor1}, a one-dimensional example is illustrated in Figure \ref{fig:theorem_illustration}. Consider the true probability function, $p(x)=0.4\exp(-1.2x)\cos(3.5\pi x)+0.4$, which is represented by a black dotted line, and the training set that contains 12 evenly-spaced inputs and the corresponding binary outputs represented by red dots. The blue line is the MMSPE predictor constructed by equation (\ref{cor:mean}) and the gray region is the corresponding 95\% confidence band constructed by the 2.5\%- and 97.5\%-quantiles. 
It appears that the proposed predictor and the confidence band reasonably capture the underlying probability.

\begin{figure}[h]
\centering
\includegraphics[width=0.6\textwidth]{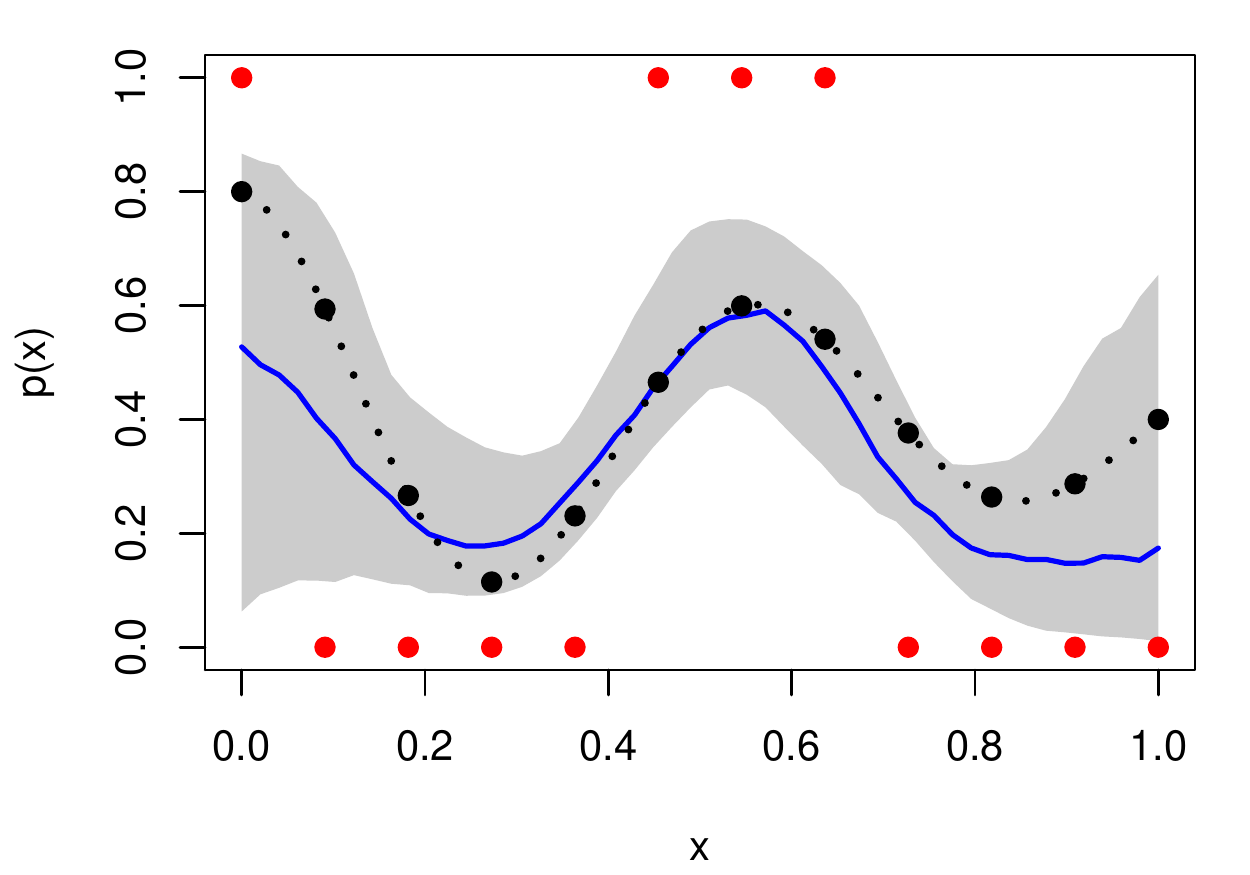}
\caption{Illustration of predictive distribution. Black dotted line represents the true probability function, red dots represent the binary response data, black dots represent the true probabilities at the chosen locations, and the emulator is represented by the blue line, with the gray shaded region providing a pointwise 95\% confidence band.}
\label{fig:theorem_illustration}
\end{figure}

%\begin{figure}
%\centering
%\includegraphics[width=0.6\textwidth]{emulation}
%\caption{}
%\label{fig:theorem_illustration}
%\end{figure}

When the historical time series for an untried setting (i.e., $y_1(\mathbf{x}_{n+1}),\ldots,y_{s-1}(\mathbf{x}_{n+1})$ in Theorem \ref{cor1}) is not available, we can emulate a completely new time series (or batch of time series) with input $\mathbf{x}_{n+1}$. The idea is to generate draws from the conditional distribution $p_s(\mathbf{x}_{n+1})|\boldsymbol{Y}$ for future outputs, starting from $s=0$, and take pointwise median of the random draws. This idea is similar to the dynamic emulators introduced by \cite{liu2009dynamic} for continuous outputs. 
The random samples from $p_s(\mathbf{x}_{n+1})|\boldsymbol{Y}$ can be generated by the fact $f(p_s(\mathbf{x}_{n+1}),\boldsymbol{p}|\boldsymbol{Y})=f(\boldsymbol{p}|\boldsymbol{Y})f(p_s(\mathbf{x}_{n+1})|\boldsymbol{p},\boldsymbol{Y})$, where $f(p_s(\mathbf{x}_{n+1})|\boldsymbol{p},\boldsymbol{Y})$ is a logit-normal distribution provided in Lemma \ref{thm1}. As mentioned above, the random samples from $f(\boldsymbol{p}|\boldsymbol{Y})$ can be generated through the MH algorithm. Therefore, generating a draw from $p_s(\mathbf{x}_{n+1})|\boldsymbol{Y}$ consists of two steps: (i) generating the ``previous'' probability values $\boldsymbol{p}^*$ given output $\boldsymbol{Y}$ from the distribution $\boldsymbol{p}|\boldsymbol{Y}$ through the MH algorithm, and (ii) based on the sample $\boldsymbol{p}^*$, draw a sample $p^*_s(\mathbf{x}_{n+1})$ from $p_s(\mathbf{x}_{n+1})|\boldsymbol{p}^*,\boldsymbol{Y}$, which is a logit-normal distribution, and also draw a sample $y^*_s(\mathbf{x}_{n+1})$ from a Bernoulli distribution with parameter $p^*_s(\mathbf{x}_{n+1})$. An explicit algorithm is given in the supplementary material S8.

\section{Simulation Studies}
%In Section \ref{sec:simulation_estimation}, simulations are conducted to demonstrate the finite sample performance of the proposed method in terms of  estimation and prediction.  
%In Section 5.2, the proposed method is compared with several existing methods by simulations generated from a modified Friedman function \citep{friedman1991multivariate}. 

In Section \ref{sec:simulation_estimation}, we conduct simulations generated from Gaussian processes to demonstrate the estimation performance. In Section \ref{sec:simulation_prediction}, the prediction performance is examined by comparing several existing methods using the data generated from a modified Friedman function \citep{friedman1991multivariate}.

\subsection{Estimation Performance}\label{sec:simulation_estimation}
 
%All the simulations are carried out on a server with Dual Xeon E5-2670 V3 Processors at 2.3 GHz CPU, and 256GB of RAM. 
Consider a 5-dimensional input space, $d=5$, and the input $\mathbf{x}$ is randomly generated from a regular grid on $[0,1]^5$. The binary output, $y_{t}(\mathbf{x})$ at time $t$, is simulated by a Bernoulli distribution with probability $p_t(\mathbf{x})$ calculated by \eqref{eq:model_GPBiTS} and $\alpha_0=0.5$, $\boldsymbol{\alpha}=(-3,2,-2,1,0.5)'$, $\varphi_1=0.8$, $\sigma^2=1$, and the power exponential correlation function \eqref{eq:powercorrelationfunction} is chosen with $\boldsymbol{\theta}=(0.5, 1.0, 1.5,  2.0, 2.5)'$ and $p=2$. Four sample size combinations of $n$ and $T$ are considered in the simulations. 

 %which addresses the identifiability issues in conventional GP modeling. 

The potential confounding between the polynomials in the mean function and the zero-mean Gaussian process can lead to the lack of identifiability, which will cause the estimated mean model to lose interpretability. In order to tackle this problem, \cite{plumlee2016orthogonal} proposed an orthogonal Gaussian process model whose stochastic part is orthogonal to the mean function. The key idea is to construct the correlation function that achieves the orthogonality. The orthogonal correlation function derived from the exponential correlation functions with power $p=2$ is given in equation (8) of \cite{plumlee2016orthogonal}. We implemented the orthogonal correlation function (abbreviated as OGP) as well as the power exponential correlation function (abbreviated as PE) in the simulation.

%The orthogonal correlation function is derived from the power exponential correlation function. 
%Details can be found in Section 3 of \cite{plumlee2016orthogonal}.

%The proposed method is implemented using two types of correlation functions. One is the power exponential correlation function, denoted by ``PE'', and the other is the orthogonal correlation function derived by the same exponential correlation function, denoted by ``OGP'', which addresses the identifiability issue in conventional GP modeling \citep{plumlee2016orthogonal}. 

The estimation results for the linear function coefficients are summarized in Table \ref{tab:simulation_estimation} based on 100 replicates for each sample size combination. In general, the proposed approach can estimate the linear function coefficients ($\alpha_0, \boldsymbol{\alpha}, \varphi_1$)  reasonably well. Compared with PE, the estimation improvement using OGP is reported as IMP. It appears that the estimation accuracy can be further improved by the use of orthogonal correlation functions. Therefore, \textit{orthogonal correlation functions are generally recommended} when estimation is of major interest, such as in variable selection and calibration problems.

\begin{table}[!h]
\centering 
{\begin{tabular}{cc|c|ccccccc}
\toprule
$n$ & $T$ & Corr & $\hat{\alpha}_0$ & $\hat{\alpha}_1$ & $\hat{\alpha}_2$ & $\hat{\alpha}_3$ & $\hat{\alpha}_4$ & $\hat{\alpha}_5$ & $\hat{\varphi}_1$ \\ 
\midrule
\multirow{5}{*}{200} & \multirow{5}{*}{20} &   \multirow{2}{*}{PE} & 0.46 & $-2.71$ & 1.82 & $-1.82$ & 0.91 & 0.46 & 0.72  \\
 & & & (0.11) & (0.15) & (0.13) & (0.12) & (0.09) & (0.10) & (0.11)\\
%\midrule
 &  & \multirow{2}{*}{OGP} &  0.48 & $-2.77$ & 1.84 & $-1.85$ & 0.91 & 0.46 & 0.71 \\
 & & & (0.12) & (0.12) & (0.10) & (0.09) & (0.08) & (0.08) & (0.10)\\
 & & IMP (\%) & 4 &2 &1& 1.5& 0&0& $-1.25$\\
\midrule
\multirow{5}{*}{200} & \multirow{5}{*}{50} & \multirow{2}{*}{PE} &   0.45 & $-2.68$ & 1.80 & $-1.79$ & 0.90 & 0.46 & 0.70 \\
 & & & (0.07) & (0.10) & (0.09) & (0.08) & (0.07) & (0.06) & (0.07)\\ 
%\midrule
  &   & \multirow{2}{*}{OGP} & 0.47 & $-2.75$ & 1.83 & $-1.83$ & 0.92 & 0.46 & 0.71 \\
 & & &(0.08) & (0.09) & (0.06) & (0.06) & (0.06) & (0.05) & (0.08)\\
 & & IMP (\%)& 4 &2.3 &1.5& 2& 2 &0& 1.25\\
\midrule
\multirow{5}{*}{500} & \multirow{5}{*}{20}& \multirow{2}{*}{PE}   & 0.47 & $-2.76$ & 1.82 & $-1.83$ & 0.91 & 0.48 & 0.73 \\
 & & & (0.09) & (0.14) & (0.12) & (0.10) & (0.08) & (0.07) & (0.07)\\ 
%\midrule
  &   &  \multirow{2}{*}{OGP}   &  0.49 & $-2.81$ & 1.89 & $-1.88$ & 0.95 & 0.46 & 0.74 \\
 & & & (0.08) & (0.09) & (0.07) & (0.06) & (0.06) & (0.06) & (0.07)\\
  & & IMP (\%)& 4 &1.7 &3.5& 2.5& 4 &4 & 1.25\\
\midrule
\multirow{5}{*}{500} & \multirow{5}{*}{50} &\multirow{2}{*}{PE}  & 0.45 & $-2.75$ & 1.83 & $-1.83$ & 0.92 & 0.45 & 0.74 \\
 & & & (0.06) & (0.07) & (0.06) & (0.06) & (0.06) & (0.05) & (0.04)\\ 
  &   & \multirow{2}{*}{OGP}  & 0.47 & $-2.80$ & 1.87 & $-1.87$ & 0.93 & 0.47 & 0.75 \\
 & & & (0.06) & (0.05) & (0.04) & (0.04) & (0.03) & (0.03) & (0.04)\\    & & IMP (\%)& 4 &1.7 &2& 2& 1 &2 & 1.25\\
\bottomrule
\end{tabular}}
\caption{Estimation of linear coefficients. The values are the average estimates over 100 replicates, while the values in parentheses are the standard deviation of the estimates. The parameter settings are $\alpha_0=0.5,\alpha_1=-3,\alpha_2=2,\alpha_3=-2,\alpha_4=1,\alpha_5=0.5$, and $\varphi_1=0.8$.} 
\label{tab:simulation_estimation}
\end{table}

The parameter estimation results for $\sigma^2$ and PE correlation parameters $\boldsymbol{\theta}$ are reported in Table \ref{tab:simulation_estimation_variance}. The estimation with OGP has similar results, so we omit them to save space.
The proposed approach tends to overestimate the correlation parameters for small sample size. This is not surprising because the estimation of correlation parameters is more challenging and the same phenomenon is observed in conventional GP models (see \cite{li2005analysis}). 
%In general, the estimation performance is improved when the sample size increases.
This problem can be ameliorated by the increase of sample size as shown in Table \ref{tab:simulation_estimation_variance}. Given the same number of total sample size, ($n=200, T=50$) and ($n=500, T=20$), it appears that a larger $n$ can improve the estimation accuracy more effectively. % when the input variables play a significant role in the underlying system. 

\begin{table}[!h]
\centering
\begin{tabular}{cc|cccccc}
\toprule
$n$ & $T$  & $\hat{\theta}_1$ & $\hat{\theta}_2$ & $\hat{\theta}_3$ & $\hat{\theta}_4$ & $\hat{\theta}_5$ & $\hat{\sigma}^2$ \\ 
\midrule
\multirow{2}{*}{200} & \multirow{2}{*}{20} &  0.86 & 1.80 & 2.35 & 3.30 & 4.10 & 0.82 \\
 & & (0.81) & (1.13) & (1.41) & (1.76) & (1.85) &  (0.07)\\
\midrule
\multirow{2}{*}{200} & \multirow{2}{*}{50} &  0.65 & 1.55 & 2.38 & 3.01 & 3.80 & 0.79 \\
 & & (0.16) & (0.63) & (1.12) & (1.25) & (1.49) &  (0.05)\\ 
\midrule
\multirow{2}{*}{500} & \multirow{2}{*}{20} & 0.61 & 1.17 & 1.93 & 2.66 & 3.24 &  0.87 \\
 & & (0.16) & (0.25) & (0.54) & (0.96) & (1.17) &  (0.05) \\ 
\midrule
\multirow{2}{*}{500} & \multirow{2}{*}{50} & 0.57 & 1.16 & 1.78 & 2.37 & 3.11 & 0.87 \\
 & & (0.08) & (0.16) & (0.35) & (0.39) & (0.68) & (0.03)\\ 
\bottomrule
\end{tabular}
\caption{Estimation of correlation parameters and variance. The values are the average estimates over 100 replicates, while the values in parentheses are the standard deviation of the estimates. The parameter settings are $\theta_1=0.5,\theta_2=1.0,\theta_3=1.5,\theta_4=2,\theta_5=2.5$, and $\sigma^2=1$.}
\label{tab:simulation_estimation_variance}
\end{table}

Based on the construction of predictive distribution in Section 4, we can emulate a new time series with an untried input. Here we generate 100 random untried inputs to examine its prediction performance. The prediction performance is evaluated by the following two measures. 
%To demonstrate the prediction performance, we consider the following two criteria. 
Define the 100 random untried inputs ($n_{\text{test}}=100$) by $\mathbf{x}^*_1,\ldots,\mathbf{x}^*_{100}$.
%the misclassification rate (MR) is calculated by
%\[
%\frac{1}{n_{\text{test}}T}\sum^{n_{\text{test}}}_{i=1}%\sum^T_{t=1}(y_t(\mathbf{x}^*_i)-\hat{y}_t(\mathbf{x}^%*_i))^2,
%\]
%where $\hat{y}_t(\mathbf{x}^*_i)$ is the predictive binary response by the proposed method. 
Since the underlying probabilities are known in the simulation settings, we evaluate the prediction performance by the root mean squared prediction error 
\[
\text{RMSPE}=\left(\frac{1}{n_{\text{test}}T}\sum^{n_{\text{test}}}_{i=1}\sum^T_{t=1}(p_t(\mathbf{x}^*_i)-\hat{p}_t(\mathbf{x}^*_i))^2\right)^{1/2},
\]
where $\hat{p}_t(\mathbf{x}^*_i)$ is the predictive probability.
The results are given in Table \ref{tab:simulation_RMSE}. Overall, the proposed predictor has the root mean squared prediction error less than $0.12$. Also, with the increase of sample size, the prediction error decreases in general.

\begin{table}[h]
\centering
\begin{tabular}{c|cccc}
\toprule
&$n=200$ & $n=200$  & $n=500$ & $n=500$ \\ 
&$T=20$ & $T=50$ & $T=20$ & $T=50$ \\ 
\midrule
%\multirow{2}{*}{MR (\%)} & 17.22 & 17.27   & 17.08 &  16.83\\
% & (1.55) & (1.39) & (1.58) & (1.19)\\
%\midrule 
\multirow{2}{*}{RMSPE} & 0.1188 & 0.1193 &0.1058  &  0.1053 \\
 & (0.0066) & (0.0058) & (0.0060) & (0.0045)\\
\bottomrule 
\end{tabular}
\caption{The comparison of RMSPEs.}
\label{tab:simulation_RMSE}
\end{table}

%By using the theoretical derivations in Section 3, 
Furthermore, the predictive distributions can be used to quantify the prediction uncertainty. The predictive distributions with two random untried inputs are shown in Figure \ref{pred}, where the green dotted lines represent the true probability, the red dashed lines represent the MMSPE predictors obtained in Theorem \ref{cor1}. From Figure \ref{pred}, it appears that the MMSPE predictors provide accurate predictions in both cases. Moreover, the predictive distributions provide rich information for statistical inference. For example, we can construct $95\%$ predictive confidence intervals for the two untried settings as indicated in blue in Figure \ref{pred}.

\begin{figure}[h]
\centering 
\includegraphics[width=0.48\textwidth]{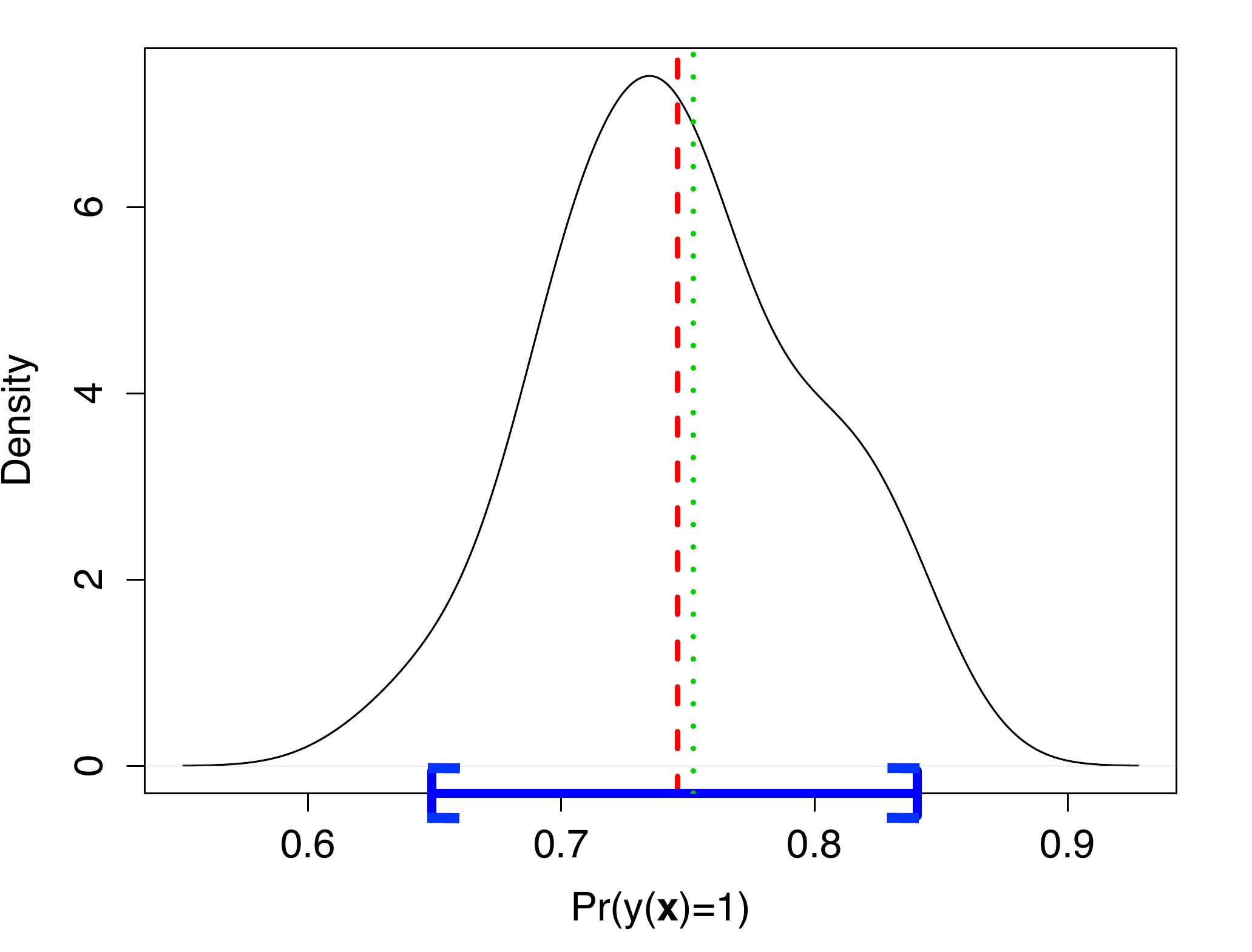}
\includegraphics[width=0.48\textwidth]{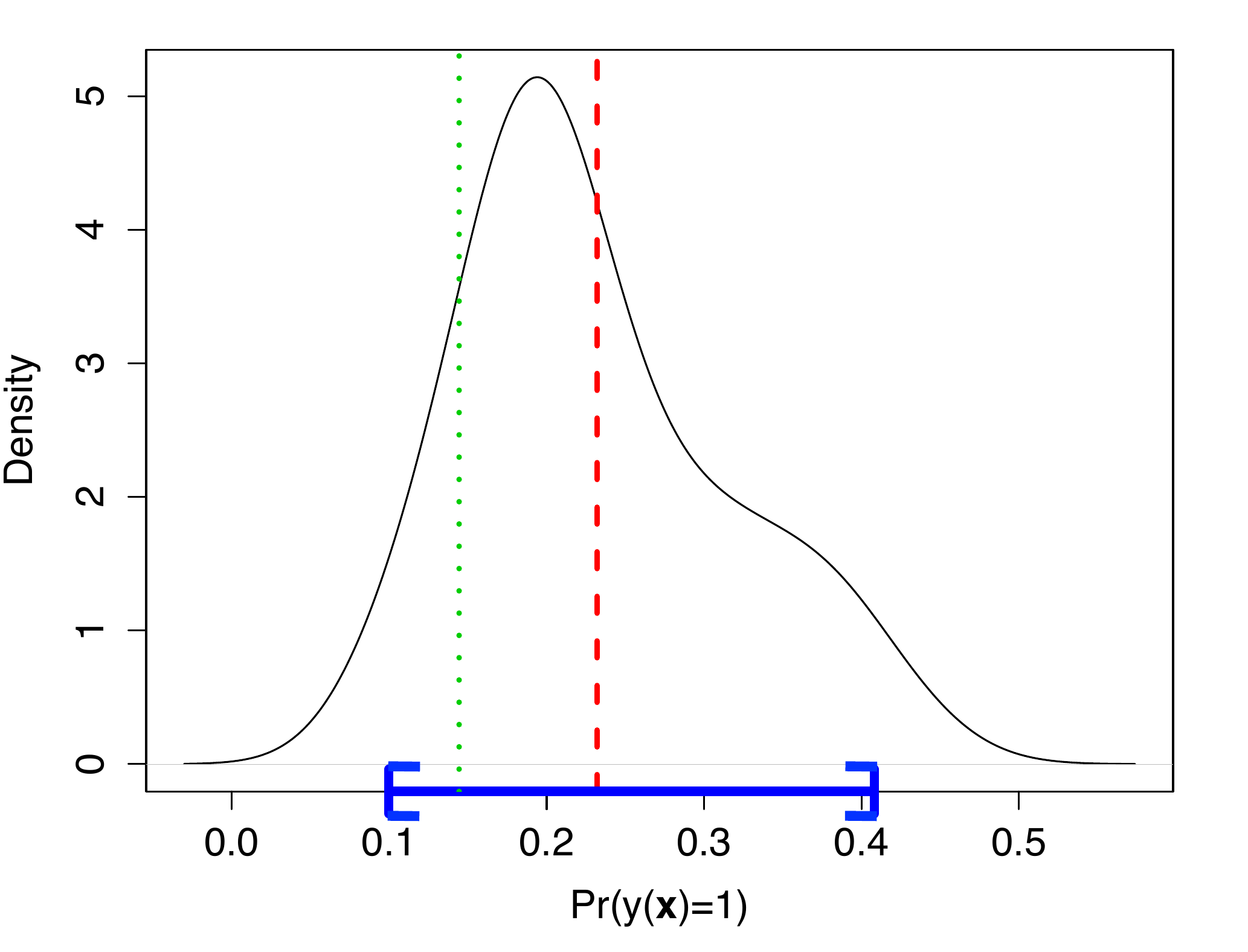}
\caption{Predictive distributions. The green dotted lines are the true probability, the red dashed lines are the MMSPE predictors, and the 95\% predictive confidence intervals are indicated in blue.}\label{pred}
\end{figure}

\subsection{Prediction Performance}\label{sec:simulation_prediction}

To examine the performance of the proposed model as an emulator, we compare its prediction accuracy with four existing methods: (1) the logistic regression model,  (2) a combination of logistic regression model with time series mean function, (3) the Bayesian generalized Gaussian process model \citep{williams1998bayesian}, which incorporates a Gaussian process prior but does not take into account the time series structure, and (4) the functional Gaussian process proposed by \cite{shi2011gaussian}, which captures the serial correlation by functional data analysis techniques. These methods are respectively implemented by \texttt{R} \citep{R} using packages \texttt{binaryGP} \citep{CL2017}, \texttt{stat} \citep{R}, a modification of \texttt{stat}, \texttt{kernlab} \citep{Rkernlab} and \texttt{GPFDA} \citep{RGPFDA} adapted to classification. 

%To examine the performance of the proposed model as an emulator, we compare its prediction accuracy with four existing methods. The first is a naive alternative using logistic regression. The second is a combination of logistic regression with time series mean function to demonstrate the contribution of Gaussian process in the proposed method. The third model is the Bayesian generalized Gaussian process model introduced by \cite{williams1998bayesian} which incorporates a Gaussian process as a prior but does not take into account the time series structure. The last one is the functional Gaussian process proposed by \cite{shi2011gaussian} which captures the serial correlation by functional data analysis techniques. These methods are implemented by \texttt{R} using packages \texttt{binaryGP}, \texttt{glm}, a modification of glm denoted by \texttt{glm-ts}, \texttt{kernlab} and \texttt{GPFDA} adapted to classification. 

%The logistic regression model with time-series is implemented via the estimation in Section \ref{sec:estimation} without the variance component, and the prediction by the similar idea of dynamic emulation in \cite{liu2009dynamic}. For logistic regression models and Bayesian generalized Gaussian process models which have no time-series mean function, we implement them by assuming $y_t$ varies independently over time. That is, a model is fitted and makes prediction at each time $t$. \textbf{(May need rewording here.)}

The simulated data are generated by a modification of the Friedman function \citep{friedman1991multivariate},
\[
{\rm{logit}}(p_t(\mathbf{x}))=y_{t-1}(\mathbf{x}) + \frac{1}{3}\left[10\sin(\pi x_1x_2)+20(x_3-0.5)^2+10x_4+5x_5\right]-5,
\]
where $\mathbf{x}\in [0,1]^5$ and the Friedman function is given in the brackets with intercept $-5$ and scale $1/3$ to ensure $p_t(\mathbf{x})$ is uniformly located at $[0,1]$.
The input $\mathbf{x}$ is randomly generated from $[0,1]^5$ and the corresponding binary output $y_t(\mathbf{x})$ at time $t$ is generated by a Bernoulli distribution with probability $p_t(\mathbf{x})$.
The size of the training data is set to be $n=200,T=20$.

Since the underlying probabilities are known in this simulation setting, the prediction performance is evaluated by RMSPE using 100 randomly generated untried settings $(n_{\rm test}=100, T=20)$. The results for the five methods based on 100 replicates are shown in the left panel of Figure \ref{fig:simulation_prediction}. In general, the proposed method has lower RMSPE than the other four methods. By incorporating a Gaussian process to model the nonlinearity, the proposed method outperforms the straightforward combination of logistic regression model and time series structure. On the other hand, comparing with Bayesian generalized Gaussian process model (i.e.,  \texttt{kernlab} in Figure \ref{fig:simulation_prediction}), the proposed method further improves the prediction accuracy by taking into account the time series structure. 
The computation time, for model fitting and prediction, is given in the right panel of Figure \ref{fig:simulation_prediction}. The proposed method is faster than \texttt{GPFDA}. Comparing with \texttt{glm} and \texttt{glm\_ts}, the major computational difficulty lies in the estimation of correlation parameters, which is a common issue in conventional GPs because there is no analytical solution for the parameter estimation. The  \texttt{kernlab} has better computational performance since it assumes that all the correlation parameters are the equal and estimated by analytic approximation, that is, $\theta_1=\ldots,\theta_d$ in \eqref{eq:powercorrelationfunction}. However, the computation time of \texttt{kernlab} is expected to increase if this assumption is relaxed to a correlation function with different correlation parameters as in \eqref{eq:powercorrelationfunction}.

\begin{figure}[h]
\centering 
\includegraphics[width=0.48\textwidth]{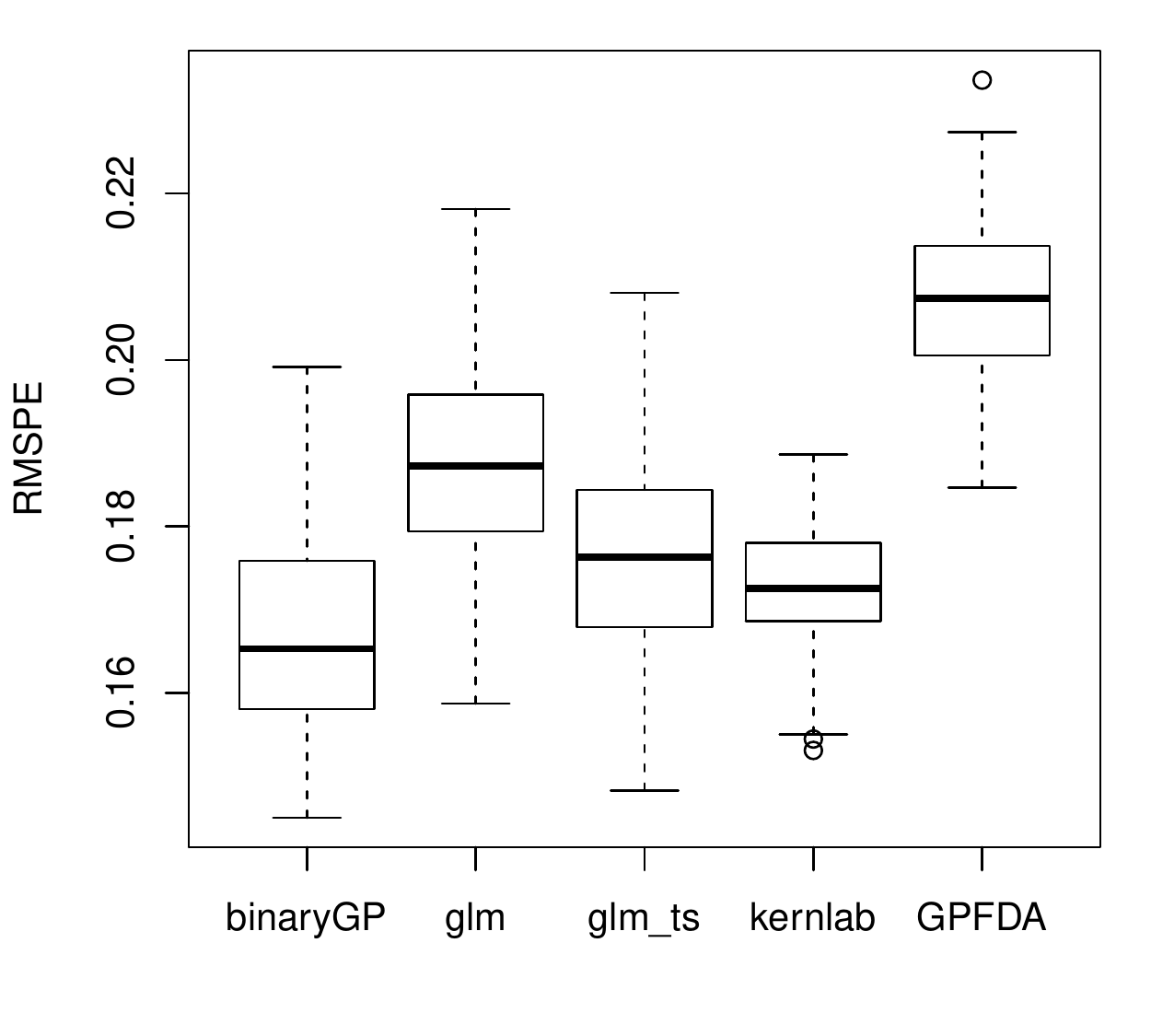}
\includegraphics[width=0.48\textwidth]{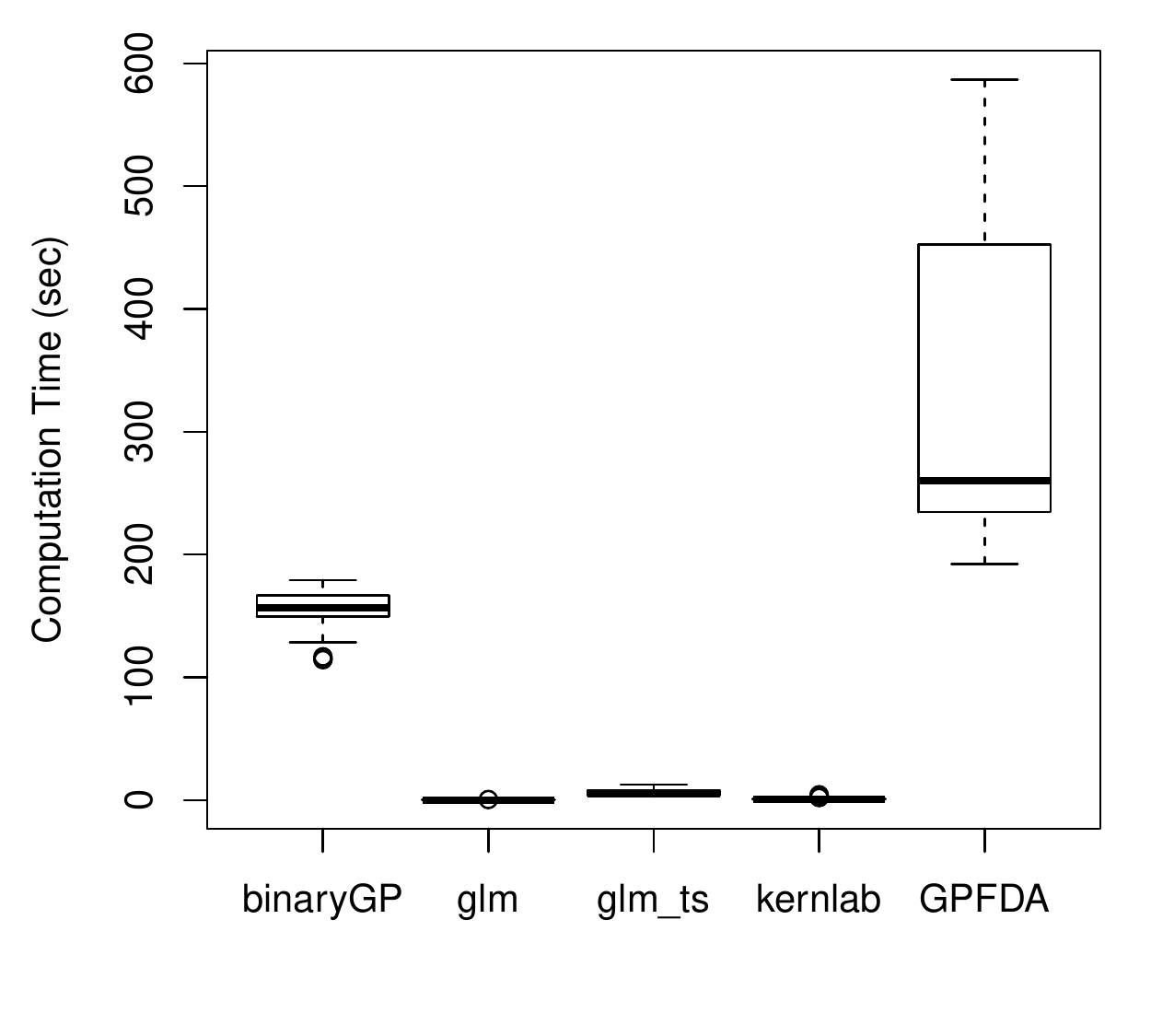}
\caption{Comparison of prediction performance in terms of accuracy (left) and computation time (right). binaryGP: proposed method, glm: logistic regression, glm\_ts: logistic regression with time-series mean function, kernlab: Bayesian generalized GP, and GPFDA: functional Gaussian process model.}\label{fig:simulation_prediction}
\end{figure}

\section{Computer Experiments for Cell Adhesion Frequency Assay}
In an earlier study based on \textit{in vitro} experiments, an important memory effect was discovered in the repeated adhesion experiments of the micropipette adhesion frequency assay. However, only limited variables of interest can be studied in the lab because of the technical complexity of the biological setting and the complicated experimental manipulation. Therefore, computer simulation experiments are performed to examine the complex mechanisms behind repeated receptor-ligand binding to rigorously elucidate the biological mechanisms behind the memory effect.  

In these computer experiments, two surfaces are simulated to reflect the two opposing membranes in the adhesion frequency assays. The molecules on the surfaces are permitted to interact for the contact duration and then separated for a period of waiting time to simulate the retract-approach phase of the assays. 
The computer experiments are constructed based on a kinetic proofreading model for receptor modification and solved through a Gillespie algorithm \citep{gillespie1976general}, which is a stochastic simulation algorithm. The contact is scored as 1 or 0 depending on whether at least one bond or no bond is observed, respectively. The process is repeated until the given number of contacts is completed. 
  
The biological system investigated here is the T Cell Receptor (TCR) binding to antigen peptide bound to Major Histocompatibility Complex (pMHC), which has previously been shown to exhibit memory in repeated contacts \citep{Zarnitsyna2007PNAS}. The TCR is the primary molecule involved in detecting foreign antigens which are presented on pMHC molecules expressed by infected cells. Memory in serial interactions of these foreign antigens may be a mechanism which underlies the major properties of T cell antigen recognition: sensitivity, specificity, and context discrimination. It has largely remained uninvestigated due to the small time scales at which the mechanism operates and the complexity of the experimental system. Although there are many possible cellular mechanisms which may induce this behavior, we investigate a specific mechanism, called free induction mechanism \citep{huang2010kinetics}, in this study as to how this memory may be controlled: pMHC binding to a single TCR within a cluster upregulates the kinetics of all TCRs within that cluster.
%1) pMHC binding to a single TCR within a cluster upregulates the kinetics of all TCRs within that cluster, or 2) pulling by engaged pMHC of a single TCR within a cluster induces upregulated kinetics to the cluster of TCRs.  

%\begin{figure}[h]
%\centering
%\includegraphics[width=0.75\textwidth]{mechanisms}
%\caption{Two cellular mechanisms investigated %in the study: (A) free induction mechanism %and (B) forced induction mechanism.}
%\label{fig:mechanism}
%\end{figure}

The free induction mechanism has six control variables given in Table \ref{realdata}. The range of each control variable in Table \ref{realdata} is given by physical principles or estimated through similar molecular interactions. The design of the control variables is a 60-run OA-based Latin hypercube designs \citep{Tang1993}.  
For each run, it consists of 50 replicates and each replicate has 100 repeated contacts ($T=100$).
 
\begin{table}[h] 
\begin{tabular}{c|c|c}  
\toprule
Variable & Description & Range\\
\midrule
$x_{K_{f,p}}$ & on-rate enhancement of activated TCRs & (1,100)\\
$x_{K_{r,p}}$ & off-rate enhancement of activated TCRs & (0.1,100)\\
$x_{T_{half}}$ & half-life of cluster activation&(0.1,10)\\
$x_{Tc}$ & cell-cell contact time& (0.1,10)\\
$x_{Tw}$ & waiting time in between contacts&(0.1,10)\\
$x_{Kc}$ & kinetic proofreading modification rate for activation of cluster & (0.1,10)\\
\bottomrule
\end{tabular} \caption{Control variables in cell adhesion frequency assay experiments.}
  \label{realdata}
\end{table}

The proposed estimation method is implemented with orthogonal correlation functions derived from the power exponential correlation function with $p=2$, which can be found in equation (8) in \cite{plumlee2016orthogonal}. We start with a large model in which the mean function includes all the main effects of the control variables and their interactions with the past time series output $y_{t-1}$. 
The model is written as: 
\begin{align*}
\text{logit}(p_{t}(\mathbf{x}))=-0.07+&\hat{\varphi}_1y_{t-1}(\mathbf{x})+\hat{\alpha}_1{x}_{K_{f,p}}+\hat{\alpha}_2{x}_{K_{r,p}}+\hat{\alpha}_3{x}_{T_{half}}+
\hat{\alpha}_4{x}_{Tc}+\hat{\alpha}_5{x}_{Tw}+\hat{\alpha}_6{x}_{Kc}+\\
%&\hat{\gamma}_1{x}_{K_{f,p}}y_{t-1}+\hat{\gamma}_2{x}_{K_{r,p}}y_{t-1}+\hat{\gamma}_3{x}_{Kc}y_{t-1}+\hat{\gamma}_4{x}_{T_{half}}y_{t-1}+\hat{\gamma}_5{x}_{Tc}y_{t-1}+\hat{\gamma}_6{x}_{Tw}y_{t-1}+Z_t(\mathbf{x}),
&(\hat{\gamma}_1{x}_{K_{f,p}}+\hat{\gamma}_2{x}_{K_{r,p}}+\hat{\gamma}_3{x}_{T_{half}}+\hat{\gamma}_4{x}_{Tc}+\hat{\gamma}_5{x}_{Tw}+\hat{\gamma}_6{x}_{Kc})y_{t-1}(\mathbf{x})+Z_t(\mathbf{x}),
\end{align*}
where all the control variables are standardized to $[0,1]$, $\hat{\sigma}=0.43$ and the estimated correlation parameters are $\hat{\boldsymbol{\theta}}=(\hat{\theta}_{K_{f,p}},\hat{\theta}_{K_{r,p}},\hat{\theta}_{T_{half}},\hat{\theta}_{Tc},\hat{\theta}_{Tw},\hat{\theta}_{Kc})=(3.28 , 1.70, 7.77 , 0.06 , 4.78, 0.74)$.
Estimation results for the mean function coefficients are given in Table \ref{table:estimation1} with p values calculated based on the asymptotic results in Theorem \ref{thm:asy_fixed}. We use these p values to perform variable selection and identify significant effects for the mean function. According to Table \ref{table:estimation1}, ${x}_{T_{half}}$ has no significant effect in the mean function at the $0.01$ level. By removing ${x}_{T_{half}}$, the model can be updated as
\begin{align*}
\text{logit}(p_{t}(\mathbf{x}))=&-0.07+0.14y_{t-1}(\mathbf{x})+0.13{x}_{K_{f,p}}+0.37{x}_{K_{r,p}}+0.47{x}_{Tc}-0.08{x}_{Tw}+0.15{x}_{Kc}\\
&+(0.16{x}_{K_{f,p}}+0.23{x}_{K_{r,p}}-0.09{x}_{T_{half}}+0.23{x}_{Tc}-0.17{x}_{Tw}+0.36{x}_{Kc})y_{t-1}(\mathbf{x})+Z_t(\mathbf{x}),
\end{align*}
where $\hat{\sigma}=0.44$, the estimated correlation parameters are $\hat{\boldsymbol{\theta}}=(\hat{\theta}_{K_{f,p}},\hat{\theta}_{K_{r,p}},\hat{\theta}_{T_{half}},\\\hat{\theta}_{Tc},\hat{\theta}_{Tw},\hat{\theta}_{Kc})=(3.27, 1.71, 7.77 , 0.06 , 4.81, 0.74)$. Based on the updated model, among the interaction effects, all the control variables except ${x}_{T_{half}}$ are significant for inducing memory in the free induction model.

\begin{table}[h]
\centering
\begin{tabular}{c|rcrc}
\toprule
   &  & Standard & &\\
 & Value &  deviation & $Z$ score & p value \\ 
\midrule
$\hat{\varphi}_1$ & 0.14 & 0.02 & 7.96 & 0.0000\\
$\hat{\alpha}_1$ & 0.13 & 0.02 & 5.96 & 0.0000\\
$\hat{\alpha}_2$ & 0.37 & 0.02 & 17.82 & 0.0000\\
$\hat{\alpha}_3$ & 0.00 & 0.02 & 0.08 & 0.9331\\
$\hat{\alpha}_4$ & 0.47 & 0.02 & 22.54 & 0.0000\\
$\hat{\alpha}_5$ & -0.08 & 0.02 & -3.79 & 0.0001\\
$\hat{\alpha}_6$ & 0.15 & 0.02 & 6.86 & 0.0000\\ 
$\hat{\gamma}_1$ & 0.16 & 0.03 & 5.2 & 0.0000\\
$\hat{\gamma}_2$ & 0.23 & 0.03 & 7.68 & 0.0000\\
$\hat{\gamma}_3$ & -0.09 & 0.03 & -2.92 & 0.0035\\
$\hat{\gamma}_4$ & 0.23 & 0.03 & 7.28 & 0.0000\\
$\hat{\gamma}_5$ & -0.17 & 0.03 & -5.47 & 0.0000 \\
$\hat{\gamma}_6$ & 0.36 & 0.03 & 11.83 & 0.0000\\
\bottomrule
\end{tabular} 
%\caption{Estimation results for the free induction mechanism.}
\caption{Estimation results.}
\label{table:estimation1}
\end{table}

%According to Table \ref{table:estimation}, we can construct a smaller model with the mean function containing only the significant effects.
%The final model can be written as:
%\begin{align*}
%\text{logit}(p_{t}(\mathbf{x}))=&-0.2996+0.7910{x}_{K_{f,p}}+0.4604{x}_{K_{r,p}}-0.9728{x}_{Kc}-0.3861{x}_{T_{half}}+\\
%&1.7834{x}_{Tc}-0.5735{x}_{Tw}+(0.7450 {x}_{K_{f,p}}+0.6040{x}_{Tc})y_{t-1}+Z_t(\mathbf{x}),
%\end{align*}
%where $\hat{\sigma}=0.35$ and the estimated power exponential correlation parameters are $\hat{\boldsymbol{\theta}}=(\hat{\theta}_{K_{f,p}},\hat{\theta}_{K_{r,p}},\hat{\theta}_{Kc},\hat{\theta}_{T_{half}},\hat{\theta}_{Tc},\hat{\theta}_{Tw})=(0.6, 0.01, 0.6, 0.04, 0.6, 0.6)$.  

The application of this statistical approach to the analysis of simulations and experimental data will be powerful in illuminating the unknown biological mechanism, and also informs the next round of experiments by advising future manipulations. Additionally, developments on the calibration of computer experiments based upon the proposed predictive distribution will help provide insight into the range of possible values of variables, such as the increases in kinetic rates, which are difficult to determine through existing methods due to the small time scale at which this mechanism operates and the limits of existing experimental techniques. 

Besides estimation, the proposed method also provides predictors which can serve as efficient and accurate emulators for untried computer experiments.  
The construction of emulator is an important step for future research on calibration where computer experiment outputs under the same settings of the lab experiments are required but not necessarily available.   
To assess the predictive performance, we compare the proposed method with the four existing methods discussed in Section 5 based on a 10-fold cross-validation study.
%Since there are 50 replicates for each run in the experiment, the average number of contacts is regarded as the true probability of contact for the validation data set. That is, for the testing input $\mathbf{x}^*_i$, $p_t(\mathbf{x}^*_i)=\sum^{50}_{k=1}(y_t(\mathbf{x}^*_i))_k$, where $(y_t(\mathbf{x}^*_i))_k$ represents the $k$-th replicate of $y_t(\mathbf{x}^*_i)$. 
%The prediction error measured by the misclassification rate is reported in Figure \ref{fig:simulation_prediction_realdata}, which shows that the proposed predictor has a smaller prediction error compared to the other alternatives. 
To evaluate the binary predictive performance, we consider four proper scoring rules defined in Table 1 of \cite{gneiting2007strictly}, including the Brier score, the spherical score, the logarithmic score, and the zero-one score. The results are reported in Table \ref{table:ScoringRule}, where larger values indicate better predictions. Compared to the four existing methods, the proposed method consistently has better prediction performance across the four scoring rules. The GPFDA approach appears to be the second best, but it is usually time-consuming to evaluate as we observed in Figure \ref{fig:simulation_prediction}. In addition, the 10-fold variation is relatively large as shown in Figure \ref{fig:ScoringRule} using the zero-one scoring rule.

\begin{table}[h]
\centering
\begin{tabular}{c|rrrrr}
\toprule
Scoring rule  & binaryGP & glm & glm\_ts & kernlab & GPFDA\\
\midrule
Brier & \textbf{-0.2092}&-0.2193&-0.2199&-0.2197&-0.2145\\
Spherical & \textbf{0.7632}&0.7499&0.7490&0.7491&0.7545\\
Logarithmic & \textbf{-0.6099}&-0.6310&-0.6323&-0.6323&-0.6166\\
Zero-one ($c=1/2$) & \textbf{0.3688}&0.3517&0.3542&0.3246&0.3650\\
\bottomrule
\end{tabular} 
%\caption{Estimation results for the free induction mechanism.}
\caption{Comparison of prediction performance using proper scoring rules. The values are the medians over the 10 folds. Larger value indicates better prediction. binaryGP: proposed method, glm: logistic regression, glm\_ts: logistic regression with time-series mean function, kernlab: Bayesian generalized GP, and GPFDA: functional Gaussian process model.}
\label{table:ScoringRule}
\end{table}
%\begin{figure}[h]
%\centering 
%\includegraphics[width=0.48\textwidth]{RMSPE_CC.pdf}
%\includegraphics[width=0.48\textwidth]{RMSPE_RP.pdf}
%\caption{Comparison of prediction performance for the free induction mechanism (left panel) and for the forced induction mechanism (right panel). binaryGP: proposed method, glm: logistic regression, glm\_ts: logistic regression with time-series mean function, kernlab: Bayesian generalized GP, and GPFDA: functional Gaussian process model.}\label{fig:simulation_prediction_realdata}
%\end{figure}

\begin{figure}[h]
\centering 
\includegraphics[width=0.5\textwidth]{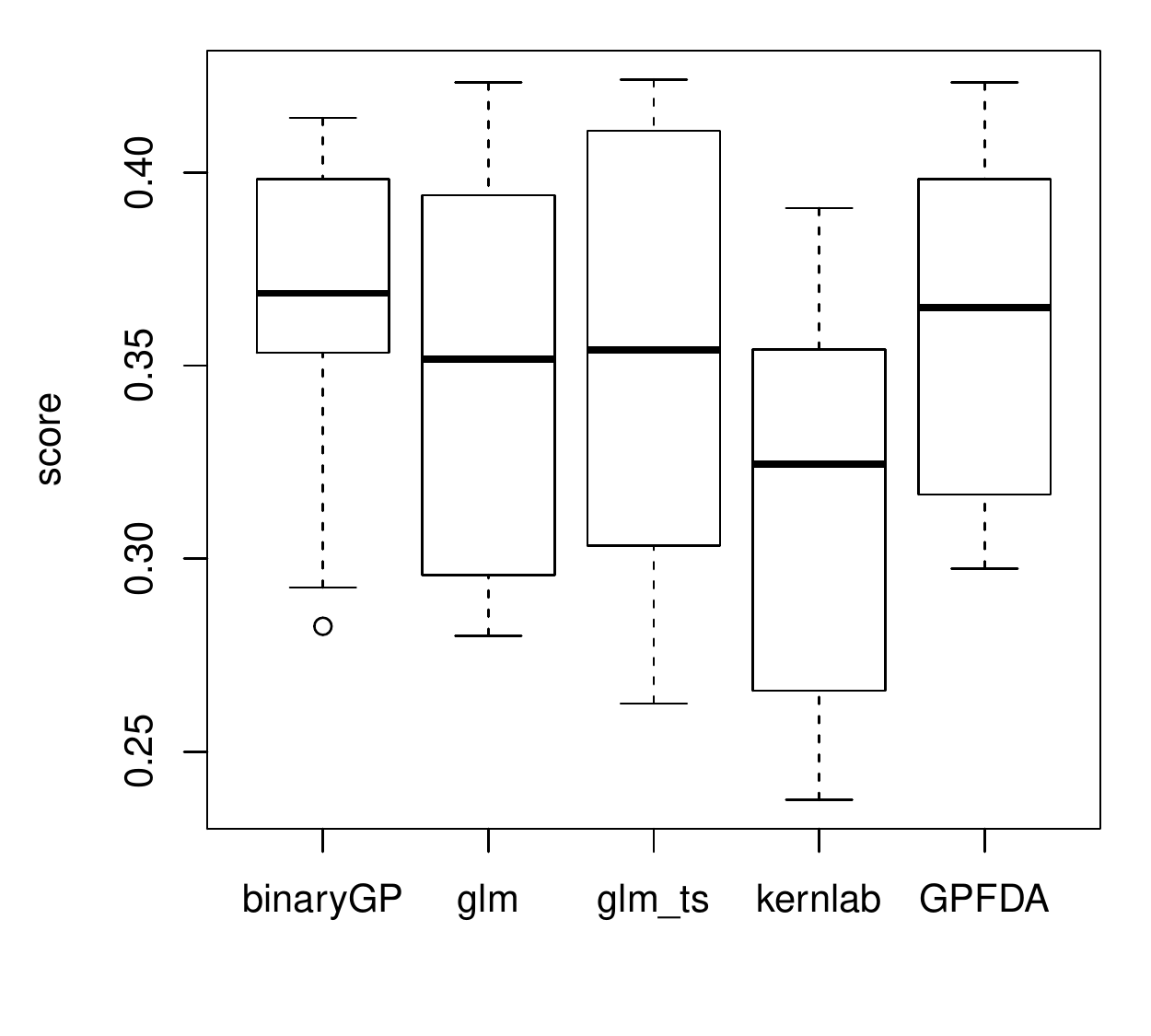}
\caption{Comparison of prediction performance using zero-one score ($c=1/2$). binaryGP: proposed method, glm: logistic regression, glm\_ts: logistic regression with time-series mean function, kernlab: Bayesian generalized GP, and GPFDA: functional Gaussian process model.}\label{fig:ScoringRule}
\end{figure}

%\begin{figure}[h]
%\centering 
%\includegraphics[width=0.48\textwidth]{MSEL_CC.pdf}
%\includegraphics[width=0.48\textwidth]{MSEL_RP.pdf}
%\caption{Comparison of mean squared error loss for the free induction mechanism (left panel) and for the forced induction mechanism (right panel). binaryGP represents the proposed method; glm represents the logistic regression model; glm\_ts represents the logistic regression model with time-series mean function; kernlab represents the Bayesian generalized GP.}\label{fig:simulation_prediction_realdata}
%\end{figure}

\section{Summary and Concluding Remarks}

In spite of the prevalence of Gaussian process models in the analysis of computer experiments, their applications are limited to the Gaussian assumption on the responses. Motivated by the study of cell adhesion where the computer simulation responses are binary time series, a generalized Gaussian process model is proposed in this paper. The estimation procedure is introduced and asymptotic properties are derived. An optimal predictor and its predictive distribution are constructed which can be used for uncertainty quantification and calibration of future computer simulations. An \texttt{R} package is available for implementing the proposed methodology. The methodology is applied to analyze stochastic computer simulations for a cell adhesion mechanism. The results reveal important biological information which is not available in lab experiments and provide valuable insights on how the next round of lab experiments should be conducted.

The current work can be extended in several directions. First, we will extend the proposed method to other non-Gaussian data, such as the count data. It is conceivable that the current estimation procedure can be directly extended to other exponential family distributions, but different predictive distributions are expected for different types of non-Gaussian data. Second, the computational cost in the proposed procedure can be further reduced. In particular, the inversion of $\boldsymbol{R}_{\boldsymbol{\theta}}$ can be computationally prohibitive when sample size is large.
This computational issue has been addressed for conventional GP models in the recent literature. Extensions of these methods (e.g., \cite{GA2015,Sea2016}) to binary responses deserve further attention. Third, many mathematical models underlying the computer simulations contain unknown parameters, which need to be estimated using data from lab experiments. This problem is called calibration and much work has been done in the computer experiment literature. However, the existing methods (e.g., \cite{kennedy2001bayesian}, \cite{tuo2015calibration} and \cite{gramacy2015calibrating}) are only applicable under the Gaussian assumption. Based upon the model and prediction procedure proposed herein, we will work on developing a calibration method for non-Gaussian data.

\vspace{5mm}
\noindent{\bf Supplementary Materials} The assumptions for Theorems \ref{thm:asy_fixed} and \ref{thm:asy_reml}, the proofs of Theorems \ref{thm:asy_fixed}, \ref{thm:asy_reml}, \ref{thm1}, \ref{thm2}, and the algorithms for estimation and emulation are given in an online supplement.

\vspace{5mm}
\if1\blind{\noindent \textbf{Acknowledgements}: The authors gratefully acknowledge helpful advice from the associate editor and two referees. This work was supported by NSF DMS 1660504 and 1660477.}
\fi

\begin{appendices}

\section{Algorithm: Estimation of $(\boldsymbol{\beta},\boldsymbol{\omega}$)}\label{append:estimationalgorithm}
\begin{algorithmic}[1]
  \State Set initial values $\boldsymbol{\omega}=(\sigma^2,\boldsymbol{\theta})=\mathbf{1}_{d+1}, \boldsymbol{\beta}=\mathbf{1}_m,p_{it}=1$, and set $\tilde{\eta}_{it}=\log\frac{p_{it}}{1-p_{it}}+\frac{y_{it}-p_{it}}{p_{it}(1-p_{it})}$ for each $i$ and $t$.
  \Repeat
     \Repeat
     	\State Set $\boldsymbol{W}$ as an $N\times N$ diagonal matrix with diagonal elements $W_{it}=p_{it}(1-p_{it})$
	 	\State Set $\boldsymbol{V}=\boldsymbol{W}^{-1}+\sigma^2(\boldsymbol{R}_{\boldsymbol{\theta}}\otimes I_T)$ 
        \State Update $\boldsymbol{\beta}=(\boldsymbol{X}'\boldsymbol{V}^{-1}\boldsymbol{X})^{-1}\boldsymbol{X}'\boldsymbol{V}^{-1}\tilde{\boldsymbol{\eta}}$
       \State Set $\boldsymbol{Z}=\sigma^2(\boldsymbol{R}_{\boldsymbol{\theta}}\otimes I_T)\boldsymbol{V}^{-1}(\tilde{\boldsymbol{\eta}}-\boldsymbol{X}'\boldsymbol{\beta})$
       \State Update $p_{it}=\left(\frac{\exp\{\boldsymbol{X}'\boldsymbol{\beta}+\boldsymbol{Z}\}}{\mathbf{1}_N+\exp\{\boldsymbol{X}'\boldsymbol{\beta}+\boldsymbol{Z}\}}\right)_{it}$ and $\tilde{\eta}_{it}=\log\frac{p_{it}}{1-p_{it}}+\frac{y_{it}-p_{it}}{p_{it}(1-p_{it})}$ for each $i$ and $t$
     \Until $\{\tilde{\eta}_{it}\}_{it}$ converges
    \State Update $\boldsymbol{\omega}=\arg\min_{\boldsymbol{\omega}}L(\boldsymbol{\omega})$, where $L(\boldsymbol{\omega})$ is the negative log-likelihood function (12)
    \State Update $(\sigma^2,\boldsymbol{\theta})=\boldsymbol{\omega}$
    \Until $\boldsymbol{\beta}$ and $\boldsymbol{\omega}$ converge
	\State Return $\boldsymbol{\beta}$ and $\boldsymbol{\omega}$
\end{algorithmic}

\section{Assumptions}\label{append:assumption}

\begin{enumerate}
\item The parameter $\boldsymbol{\beta}$ belongs to an open set $B\subseteq \mathbb{R}^m$ and the parameter $\boldsymbol{\omega}$ belongs to an open set $\Omega\subseteq \mathbb{R}^{d+1}$.
\item The model matrix $X_{it}$ lies almost surely in a nonrandom compact subset of $\mathbb{R}^m$ such that $Pr(\sum^n_{i=1}\sum^T_{t=1}X'_{it}X_{it}>0)=1$.
\end{enumerate}
For any matrix $A$, define $\|A\|\equiv\sqrt{\text{tr}(A'A)}$; for the covariance matrix $\boldsymbol{V}(\boldsymbol{\omega})$, define $V_i(\boldsymbol{\omega})\equiv\partial\boldsymbol{V}(\boldsymbol{\omega})/\partial \omega_i$ and $V_{ij}(\boldsymbol{\omega})\equiv\partial\boldsymbol{V}(\boldsymbol{\omega})/\partial \omega_i\partial\omega_j$; for $\boldsymbol{\omega}\in\Omega$, denote $\stackrel{u}{\longrightarrow}$ as uniform convergence of nonrandom functions over compact subsets of $\Omega$.
\begin{enumerate}
\setcounter{enumi}{2}
\item  $J_N(\boldsymbol{\omega})P_N(\boldsymbol{\omega})^{-1}\stackrel{d}{\longrightarrow}W(\boldsymbol{\omega})$ for some nonsingular $W(\boldsymbol{\omega})$, which is continuous in $\boldsymbol{\omega}$, where  $P_N(\boldsymbol{\omega})=\text{diag}(\|\varPi(\boldsymbol{\omega})V_1(\boldsymbol{\omega})\|,\ldots,\|\varPi(\boldsymbol{\omega})V_{d+1}(\boldsymbol{\omega})\|)$ and $\varPi(\boldsymbol{\omega})=\boldsymbol{V}(\boldsymbol{\omega})^{-1}-\boldsymbol{V}(\boldsymbol{\omega})^{-1}\boldsymbol{X}(\boldsymbol{X}'\boldsymbol{V}(\boldsymbol{\omega})^{-1}\boldsymbol{X})^{-1}\boldsymbol{X}'\boldsymbol{V}(\boldsymbol{\omega})^{-1}$.
%\item If there exists a sequence $\{r_N\}_{N\geq 1}$ with $1\leq r_N<N-m$ for all $N\geq 1$, such that
%\begin{align*}
%&(\lambda_N/\lambda_1)^4\left\{N/(N-r_N-m)^2\right\}\left\{\sum^{d+1}_{i=1}(\lambda^{i}_N/\lambda^{i}_{r_N})^2\right\}^2\stackrel{u}{\longrightarrow} 0\quad\text{and}\\
%&(\lambda^2_N/\lambda_1)^2\left\{N/(N-r_N-m)^2 \right\}^2\left\{\sum^{d+1}_{i=1}\sum^{d+1}_{j=1}(\lambda^{ij}_N)^2(\lambda^{ij}_N)^2(\lambda^{i}_{r_N}\lambda^{j}_{r_N})^{-2}\right\}\stackrel{u}{\longrightarrow} 0,
%\end{align*}
%where $|\lambda_1|\leq\ldots\leq|\lambda_N|$ are the absolute eigenvalues of $\boldsymbol{V}(\boldsymbol{\omega})$, $|\lambda^i_1|\leq\ldots\leq|\lambda^i_N|$ are the absolute eigenvalues of $V_i(\boldsymbol{\omega})$, and $|\lambda^{ij}_1|\leq\ldots\leq|\lambda^{ij}_N|$ are the absolute eigenvalues of $V_{ij}(\boldsymbol{\omega})$.
\item If there exists a sequence $\{r_N\}_{N\geq 1}$ with $\lim\sup_{N\rightarrow\infty}r_N/N\leq 1-\delta$, for some $\delta\in(0,1)$, such that for any compact subset $K\subseteq\Omega$, there exist constants $0<C_1(K)<\infty$ and $C_2(K)>0$ such that
\begin{equation*}
\lim\sup_{N\rightarrow\infty}\max\{|\lambda_N|,|\lambda^i_N|,|\lambda^{ij}_N|:1\leq i,j\leq k \}<C_1(K)<\infty
\end{equation*}
and 
\begin{equation*}
\lim\sup_{N\rightarrow\infty}\min\{|\lambda_1|,|\lambda^i_{r_N}|:1\leq i\leq k \}>C_2(K)>0,
\end{equation*}
uniformly in $\boldsymbol{\omega}\in K$, where $|\lambda_1|\leq\ldots\leq|\lambda_N|$ are the absolute eigenvalues of $\boldsymbol{V}(\boldsymbol{\omega})$, $|\lambda^i_1|\leq\ldots\leq|\lambda^i_N|$ are the absolute eigenvalues of $V_i(\boldsymbol{\omega})$, and $|\lambda^{ij}_1|\leq\ldots\leq|\lambda^{ij}_N|$ are the absolute eigenvalues of $V_{ij}(\boldsymbol{\omega})$.
\end{enumerate}

Assumption 2 holds when the row vectors of $\boldsymbol{X}$ are linear independent. Thus, if only the linear effect is considered in the mean function, then orthogonal designs or orthogonal array-based designs, such as OA-based Latin hypercube designs \citep{Tang1993}, can be chosen for sampling schemes. The conditions for Assumption 4 can be referred to \cite{cressie1996asymptotics}, in which the checkable conditions for rectangular lattice of data sites and irregularly located data sites are given. For instance, for rectangular lattice of data sites, with certain correlation functions, a sufficient condition is choosing data locations whose minimum distance is sufficiently large. More details can be seen in \cite{cressie1996asymptotics}. Thus, space-filling designs, such as Latin hypercube designs \citep{mckay1979comparison} and maximin distance designs \citep{johnson1990minimax}, can be chosen for sampling schemes.

\section{Proof of Theorem 3.1}\label{append:asy_fixed}
The model (4) can be seen as a binary time series model with random effects by multiplying an identity matrix on $\boldsymbol{Z}$, that is,
\[
\text{logit}(\boldsymbol{p})=\boldsymbol{X}\boldsymbol{\beta}+I_N\boldsymbol{Z},\quad\boldsymbol{Z}\sim\mathcal{N}(\mathbf{0}_N,\Sigma(\boldsymbol{\omega})),
\]
where $I_N$ and $\boldsymbol{Z}$ are viewed as the model matrix and coefficients of random effects, respectively. Therefore, if the variance-covariance parameters are given, inference of $\boldsymbol{\beta}$ is a special case of the binary time series model with random effects in \cite{hung2012binary}. Therefore, following Theorem 1 in \cite{hung2012binary}, the score function $S_N(\boldsymbol{\beta},\boldsymbol{\omega})$ is asymptotically normally distributed. 

\section{Proof of Theorem 3.3}\label{append:asy_reml}
According to \cite{breslow1993approximate}, one can view the inference on the variance-variance component as an iterative procedure for the linear mixed model
\[
\tilde{\boldsymbol{\eta}}=\boldsymbol{X}\boldsymbol{\beta}+I_N\boldsymbol{Z}+\epsilon,\quad \epsilon\sim\mathcal{N}(\mathbf{0}_N,\boldsymbol{W}^{-1})
\]
with the iterative weight $\boldsymbol{W}^{-1}$. Thus, it is a special case of the Gaussian general linear model in \cite{cressie1993asymptotic} with response vector $\tilde{\boldsymbol{\eta}}$ and variance-covariance component $\Sigma(\boldsymbol{\omega})+\boldsymbol{W}^{-1}$ with parameters $\boldsymbol{\omega}$. Since the asymptotic distribution of REML estimators for the variance-covariance parameters has been shown in \cite{cressie1993asymptotic} for a Gaussian general linear model, the result directly follows as a special case of Corollary 3.3 in \cite{cressie1993asymptotic}. Note that Assumption 4 in the supplementary material \ref{append:assumption} implies the conditions for Corollary 3.3 in \cite{cressie1993asymptotic}. See the proof of Theorem 2.2 in \cite{cressie1996asymptotics}.

\section{Proof of Lemma 4.1}\label{append:thm1}
We start the proof by deriving the conditional distribution from a simple model (1) (without time-series), and then extend the result to prove Lemma 4.1. First, a definition and a lemma about multivariate log-normal distribution are in order.

\begin{definition}\label{def:lognormal}
Suppose $\boldsymbol{\xi}=(\xi_1,\ldots,\xi_n)'$ has a multivariate normal distribution with mean $\boldsymbol{\mu}_n$ and covariance variance $\boldsymbol{\Sigma}_{n\times n}$. Then $\boldsymbol{b}=\exp\{\boldsymbol{\xi}\}$ has a \textit{multivariate log-normal distribution}. Denote it as $\boldsymbol{b}\sim \mathcal{LN}(\boldsymbol{\mu}_n, \boldsymbol{\Sigma}_{n\times n})$.
\end{definition}

\begin{lemma}\label{lemma1}
Suppose $\boldsymbol{b}^n$ and $b_{n+1}$ have a multivariate log-normal distribution
\[
\left(\begin{array}{c}
\boldsymbol{b}^n\\ 
b_{n+1}
\end{array}\right)\sim\mathcal{LN}\left(\left(\begin{array}{c}
\boldsymbol{\mu}^n \\ 
\mu_{n+1}
\end{array}\right) ,
\left[\begin{matrix}
\boldsymbol{\Sigma}_{n\times n} &  \boldsymbol{r} \\
\boldsymbol{r}'  & \sigma^2_{n+1} 
\end{matrix}\right] \right).
\]
The conditional distribution of $b_{n+1}$ given $\boldsymbol{b}^n$ is
$b_{n+1}|\boldsymbol{b}^n \sim \mathcal{LN}(\mu^*,v^*)$, where
$\mu^*=\mu_{n+1}+\boldsymbol{r}'\boldsymbol{\Sigma}_{n\times n}^{-1}(\log\boldsymbol{b}^n-\boldsymbol{\mu}^n)$ and $v^*=\sigma^2_{n+1}-\boldsymbol{r}'\boldsymbol{\Sigma}_{n\times n}^{-1}\boldsymbol{r}.$
\end{lemma}

\begin{proof}
Using transformation of a standard normal distribution, one can show that the joint probability density function of the multivariate log-normal distribution $\boldsymbol{b}^n$ is
\begin{equation*}\label{eq:multilognormpdf}
g_{\boldsymbol{b}^n}(b_1,\ldots,b_n)=\frac{1}{(2\pi)^{n/2}|\boldsymbol{\Sigma}_{n\times n}|^{1/2}}\frac{1}{\prod^n_{i=1}b_i}\exp\{-\frac{1}{2}\left(\log\boldsymbol{b}^n-\boldsymbol{\mu}_n\right)'\boldsymbol{\Sigma}_{n\times n}^{-1}\left(\log\boldsymbol{b}^n-\boldsymbol{\mu}_n\right)\}.
\end{equation*}
Denote $\boldsymbol{b}^{n+1}=(b_1,\ldots,b_n,b_{n+1})$, $\boldsymbol{\mu}^{n+1}=(\mu_1,\ldots,\mu_n,\mu_{n+1})$ and 
\begin{displaymath}
\boldsymbol{\Sigma}_{(n+1)\times (n+1)}= 
\left[\begin{matrix}
\boldsymbol{\Sigma}_{n\times n} &  \boldsymbol{r} \\
\boldsymbol{r}'  & \sigma_{n+1} 
\end{matrix}\right].
\end{displaymath}
Then, the conditional probability density function of $b_{n+1}$ given $\boldsymbol{b}^n$ can be derived as 
\begin{align*}
g_{b_{n+1}|\boldsymbol{b}^n}(b_{n+1}|\boldsymbol{b}^n)
	&\propto g(b_1,\ldots,b_n,b_{n+1})\\
	&\propto\frac{1}{b_{n+1}}\exp\{-\frac{1}{2}\left(\log\boldsymbol{b}^{n+1}-\boldsymbol{\mu}_{n+1}\right)'\boldsymbol{\Sigma}^{-1}_{(n+1)\times (n+1)}\left(\log\boldsymbol{b}^{n+1}-\boldsymbol{\mu}_{n+1}\right)\}.
\end{align*}
Let $\boldsymbol{a}_1=\log\boldsymbol{b}^n-\boldsymbol{\mu}^{n}$ and
$\boldsymbol{a}_2=\log b^{n+1}-\mu^{n+1}$. Applying the partitioned matrix inverse results (page 99 of \cite{harville1997matrix}) gives
\begin{align*}
&\left(\log\boldsymbol{b}^{n+1}-\boldsymbol{\mu}^{n+1}\right)'\boldsymbol{\Sigma}^{-1}_{(n+1)\times (n+1)}\left(\log\boldsymbol{b}^{n+1}-\boldsymbol{\mu}^{n+1}\right)\\
=&
\left[\begin{matrix}\boldsymbol{a}'_1 & \boldsymbol{a}'_2\end{matrix}\right]
\left[\begin{matrix}
\boldsymbol{\Sigma}_{n\times n} &  \boldsymbol{r} \\
\boldsymbol{r}'  & \sigma_{n+1} 
\end{matrix}\right]^{-1}
\left[\begin{matrix}\boldsymbol{a}_1 \\ \boldsymbol{a}_2\end{matrix}\right]\\
=&(\boldsymbol{a}_2-\boldsymbol{r}'\boldsymbol{\Sigma}^{-1}_{n\times n}\boldsymbol{a}_1)'\sigma^{-1}_{22\cdot 1}(\boldsymbol{a}_2-\boldsymbol{r}'\boldsymbol{\Sigma}^{-1}_{n\times n}\boldsymbol{a}_1)+\boldsymbol{a}'_1\boldsymbol{\Sigma}^{-1}_{n\times n}\boldsymbol{a}_1\\
=&(\boldsymbol{a}_2-\boldsymbol{r}'\boldsymbol{\Sigma}^{-1}_{n\times n}\boldsymbol{a}_1)^2/\sigma_{22\cdot 1}+\boldsymbol{a}'_1\boldsymbol{\Sigma}^{-1}_{n\times n}\boldsymbol{a}_1,
\end{align*}
where $\sigma_{22\cdot 1}=\sigma^2_{n+1}-\boldsymbol{r}'\boldsymbol{\Sigma}^{-1}_{n\times n}\boldsymbol{r}$ and is a real number.

Thus, the conditional probability density function of $b_{n+1}$ given $\boldsymbol{b}^n$ can be simplified as
\begin{align*}
g_{b_{n+1}|\boldsymbol{b}^n}(b_{n+1}|\boldsymbol{b}^n)&\propto \frac{1}{b_{n+1}}\exp\{-\frac{1}{2\sigma_{22\cdot 1}}(\boldsymbol{a}_2-\boldsymbol{r}'\boldsymbol{\Sigma}^{-1}_{n\times n}\boldsymbol{a}_1)^2-\frac{1}{2}\boldsymbol{a}'_1\boldsymbol{\Sigma}^{-1}_{n\times n}\boldsymbol{a}_1\}\\
&\propto\frac{1}{b_{n+1}}\exp\{-\frac{1}{2\sigma_{22\cdot 1}}(\boldsymbol{a}_2-\boldsymbol{r}'\boldsymbol{\Sigma}^{-1}_{n\times n}\boldsymbol{a}_1)^2\}\\
&=\frac{1}{b_{n+1}}\exp\{-\frac{1}{2\sigma_{22\cdot 1}}\left(\log b_{n+1}-(\mu_{n+1}+\boldsymbol{r}'\boldsymbol{\Sigma}^{-1}_{n\times n}(\log\boldsymbol{b}^{n}-\boldsymbol{\mu}^{n}))\right)^2\}.
\end{align*}
Therefore, according to the probability density function of a log-normal distribution, we have $b_{n+1}|\boldsymbol{b}^n \sim \mathcal{LN}(\mu^*,v^*)$, where $\mu^*=\mu_{n+1}+\boldsymbol{r}'\boldsymbol{\Sigma}^{-1}_{n\times n}(\log\boldsymbol{b}^n-\boldsymbol{\mu}^{n})$
and 
$v^*=\sigma_{22\cdot 1}=\sigma^2_{n+1}-\boldsymbol{r}'\boldsymbol{\Sigma}_{n\times n}^{-1}\boldsymbol{r}$.
\end{proof}

\begin{lemma}\label{lemma2}
Consider the model (1) (without time-series), given $(p(\mathbf{x}_1),\ldots,p(\mathbf{x}_n))'=\boldsymbol{p}^n$, the conditional distribution of $p(\mathbf{x}_{n+1})$ is a logit-normal distribution, that is,
$
p(\mathbf{x}_{n+1})|\boldsymbol{p}^n\sim Logitnormal(m(\boldsymbol{p}^n), v(\boldsymbol{p}^n))
$
with
$$m(\boldsymbol{p}^n)=\mu(\mathbf{x}_{n+1})+\boldsymbol{r}_{\boldsymbol{\theta}}'\boldsymbol{R}_{\boldsymbol{\theta}}^{-1}(\log \frac{\boldsymbol{p}^n}{\mathbf{1}-\boldsymbol{p}^n}-\boldsymbol{\mu}^n)\quad{\text{and}}\quad
v(\boldsymbol{p}^n)=\sigma^2(1-\boldsymbol{r}_{\boldsymbol{\theta}}'\boldsymbol{R}_{\boldsymbol{\theta}}^{-1}\boldsymbol{r}_{\boldsymbol{\theta}}),$$ 
where $\boldsymbol{\mu}^n=(\mu(\mathbf{x}_1),\ldots,\mu(\mathbf{x}_n))'$, $\mu(\mathbf{x}_i)=\alpha_0+\mathbf{x}'_{i}\boldsymbol{\alpha}$, $\boldsymbol{r}_{\boldsymbol{\theta}}=(R_{\boldsymbol{\theta}}(\mathbf{x}_{n+1},\mathbf{x}_1),\ldots,R_{\boldsymbol{\theta}}(\mathbf{x}_{n+1},\mathbf{x}_n))'$, and $\boldsymbol{R}_{\boldsymbol{\theta}}=\{R_{\boldsymbol{\theta}}(\mathbf{x}_i,\mathbf{x}_j)\}$.
\end{lemma}
\begin{proof}
Let $\eta_i=\mu(\mathbf{x}_i)+Z(\mathbf{x}_i)$ and $b_i=\exp\{\eta_i\}=p(\mathbf{x}_i)/(1-p(\mathbf{x}_i))$ for $i=1,\ldots,n+1$. Since $(\eta_1,\ldots,\eta_n,\eta_{n+1})'\sim\mathcal{N}(\boldsymbol{\mu}^{n+1},\sigma^2\boldsymbol{R}^*_{\boldsymbol{\theta}})$, where $\boldsymbol{\mu}^{n+1}=((\boldsymbol{\mu}^{n})',\mu(\mathbf{x}_{n+1}))'$ and 
\[
\boldsymbol{R}^*_{\boldsymbol{\theta}}=\left[\begin{matrix}
\boldsymbol{R}_{\boldsymbol{\theta}} &  \boldsymbol{r}_{\boldsymbol{\theta}} \\
\boldsymbol{r}_{\boldsymbol{\theta}}'  & 1 
\end{matrix}\right],
\] we have $(b_1,\ldots,b_n,b_{n+1})'\sim \mathcal{LN}(\boldsymbol{\mu}^{n+1}, \sigma^2\boldsymbol{R}^*_{\boldsymbol{\theta}})$ by Definition \ref{def:lognormal}. Thus, using Jacobian of the transformation and Lemma \ref{lemma1}, we have
\begin{align*}
&g_{p(\mathbf{x}_{n+1})|p(\mathbf{x}_1),\ldots,p(\mathbf{x}_n)}(p_{n+1}|p_1,\ldots,p_n)\\
=& g_{b_{n+1}|b_1,\ldots,b_n}(\frac{p_{n+1}}{1-p_{n+1}}|\frac{p_1}{1-p_1},\ldots,\frac{p_n}{1-p_n})\frac{1}{(1-p_{n+1})^2}\\
\propto& \frac{1-p_{n+1}}{p_{n+1}}\exp\{-\frac{\left(\log \frac{p_{n+1}}{1-p_{n+1}}-(\mu(\mathbf{x}_{n+1})+\boldsymbol{r}_{\boldsymbol{\theta}}'\boldsymbol{R}_{\boldsymbol{\theta}}^{-1}(\log \frac{\boldsymbol{p}^n}{\mathbf{1}-\boldsymbol{p}^n}-\boldsymbol{\mu}^n))\right)^2}{2\sigma^2(1-\boldsymbol{r}_{\boldsymbol{\theta}}'\boldsymbol{R}_{\boldsymbol{\theta}}^{-1}\boldsymbol{r}_{\boldsymbol{\theta}})}\}\frac{1}{(1-p_{n+1})^2}\\
\propto& \frac{1}{p_{n+1}(1-p_{n+1})}\exp\{-\frac{\left(\log \frac{p_{n+1}}{1-p_{n+1}}-(\mu(\mathbf{x}_{n+1})+\boldsymbol{r}_{\boldsymbol{\theta}}'\boldsymbol{R}_{\boldsymbol{\theta}}^{-1}(\log \frac{\boldsymbol{p}^n}{\boldsymbol{1}-\boldsymbol{p}^n}-\boldsymbol{\mu}^n))\right)^2}{2\sigma^2(1-\boldsymbol{r}_{\boldsymbol{\theta}}'\boldsymbol{R}_{\boldsymbol{\theta}}^{-1}\boldsymbol{r}_{\boldsymbol{\theta}})}\}.
\end{align*}
Therefore, according to the probability density function of a logit-normal distribution, we have $p(\mathbf{x}_{n+1})|\boldsymbol{p}^n\sim Logitnormal(m(\boldsymbol{p}^n), v(\boldsymbol{p}^n))$.
\end{proof}

%We start the proof by deriving the conditional distribution from a simple model \eqref{eq:simplemodel} (without time-series), and then extend the result to prove Lemma \ref{thm1}. First, a definition and a lemma about multivariate log-normal distribution are in order.

Similarly, the result of Lemma \ref{lemma2} can be extended to the general model (3). Given $\boldsymbol{Y}=(\mathbf{y}'_1,\ldots,\mathbf{y}'_T,y_{n+1,1},\ldots,y_{n+1,s-1})'$, at a fixed time-step $s$, $p_s(\mathbf{x}_i)$ can be seen to have the model (1) with mean function $\mu(\mathbf{x}_i,\boldsymbol{Y})=\sum^R_{r=1}\varphi_ry_{i,s-r}+\alpha_0+\mathbf{x}'_{i}\boldsymbol{\alpha}+\sum^L_{l=1}\boldsymbol{\gamma}_l\mathbf{x}_{i}y_{i,s-l}$. Thus, by Lemma \ref{lemma2}, denote $\boldsymbol{p}_s=(p_s(\mathbf{x}_1),\ldots,p_s(\mathbf{x}_n))'$, we have 
\[
p_s(\mathbf{x}_{n+1})|\boldsymbol{p}_s,\boldsymbol{Y}\sim Logitnormal(m(\boldsymbol{p}_s,\boldsymbol{Y}), v(\boldsymbol{p}_s,\boldsymbol{Y})),
\] where $m(\boldsymbol{p}_s,\boldsymbol{Y})=\mu(\mathbf{x}_{n+1},\boldsymbol{Y})+\boldsymbol{r}_{\boldsymbol{\theta}}'\boldsymbol{R}_{\boldsymbol{\theta}}^{-1}(\log \frac{\boldsymbol{p}_s}{\mathbf{1}-\boldsymbol{p}_s}-\boldsymbol{\mu}^n),\boldsymbol{\mu}^n=(\mu(\mathbf{x}_1,\boldsymbol{Y}),\ldots,\mu(\mathbf{x}_n,\boldsymbol{Y}))'$,
and 
$
v(\boldsymbol{p}_s,\boldsymbol{Y})=\sigma^2(1-\boldsymbol{r}_{\boldsymbol{\theta}}'\boldsymbol{R}_{\boldsymbol{\theta}}^{-1}\boldsymbol{r}_{\boldsymbol{\theta}}).$
By the fact that $Z_t(\mathbf{x})$ is independent over time, which implies $p_s(\mathbf{x})$ is independent of $p_t(\mathbf{x})$ for any $t\neq s$, $p_s(\mathbf{x}_{n+1})|D_{n+1,s}$ and $p_s(\mathbf{x}_{n+1})|\boldsymbol{p}_s,\boldsymbol{Y}$ have the same distribution. So, $p_s(\mathbf{x}_{n+1})|D_{n+1,s}\sim Logitnormal(m(D_{n+1,s}), v(D_{n+1,s}))$, where $m(D_{n+1,s})=m(\boldsymbol{p}_s,\boldsymbol{Y})$ and $v(D_{n+1,s})=v(\boldsymbol{p}_s,\boldsymbol{Y})$.

\section{Proof of Theorem 4.3}\label{append:thm2}

(i) First, one can show that if $(p_s(\mathbf{x}_{n+1}),D_{n+1,s})$ has a joint distribution for which the conditional mean of $p_s(\mathbf{x}_{n+1})$ given $D_{n+1,s}$ exists, then $\mathbb{E}\left[p(\mathbf{x}_{n+1})|D_{n+1,s}\right]$ is the minimum mean squared error predictor of $p(\mathbf{x}_{n+1})$. See Theorem 3.2.1 in \cite{santner2003design}. Thus, by the result of Lemma 4.1, we have the conditional mean $\mathbb{E}\left[p(\mathbf{x}_{n+1})|D_{n+1,s}\right]=\kappa(m(D_{n+1,s}),v(D_{n+1,s}))$ with variance $\mathbb{V}\left[p(\mathbf{x}_{n+1})|D_{n+1,s}\right]=\tau(m(D_{n+1,s}),v(D_{n+1,s}))$. 

\noindent(ii) If $\mathbf{x}_{n+1}=\mathbf{x}_i$ for $i=1,\ldots,n$, then $m(D_{n+1,s})=\log (p_s(\mathbf{x}_i)/(1-p_s(\mathbf{x}_i)))$ and $v(D_{n+1,s})=0$, which implies that $$\kappa(m(D_{n+1,s}),0)=\exp\{m(D_{n+1,s})\}/(1+\exp\{m(D_{n+1,s})\})=p_s(\mathbf{x}_i)$$ and $\tau(m(D_{n+1,s}),0)=0$ by using transformation of a normal distribution. Thus, by Theorem 4.3 (i), we have $
\mathbb{E}\left[p_s(\mathbf{x}_{n+1})|D_{n+1,s}\right]=p_s(\mathbf{x}_i)$ and $\mathbb{V}\left[p_s(\mathbf{x}_{n+1})|D_{n+1,s}\right]=0.$

\noindent(iii) Let $X\sim\mathcal{N}(m(D_{n+1,s}), v(D_{n+1,s}))$, $P=\exp\{X\}/(1+\exp\{X\})$, which has the distribution $Logitnormal(m(D_{n+1,s}), v(D_{n+1,s}))$, and $Q(q;D_{n+1,s})$ be the $q$-th quantile of $P$. Consider the function $f(x)=\log (x/(1-x))$. The derivative is $f'(x)=1/(x(1-x))$. Thus, for $0<x<1$ the derivative is positive and the $f(x)$ function is increasing in $x$. Then, 
\begin{align*}
&Pr\left\{P>Q(q;D_{n+1,s})\right\}=q\\
\Leftrightarrow&Pr\left\{\frac{\exp\{X\}}{1+\exp\{X\}}>Q(q;D_{n+1,s})\right\}=q\\
\Leftrightarrow&Pr\left\{f(\frac{\exp\{X\}}{1+\exp\{X\}})>f(Q(q;D_{n+1,s}))\right\}=q\\
\Leftrightarrow&Pr\left\{X>\log\frac{Q(q;D_{n+1,s})}{1-Q(q;D_{n+1,s})}\right\}=q\\
\end{align*}
\begin{align*}
\Leftrightarrow&Pr\left\{\frac{X-m(D_{n+1,s})}{\sqrt{v(D_{n+1,s})}}>\frac{1}{\sqrt{v(D_{n+1,s})}}\left(\log\frac{Q(q;D_{n+1,s})}{1-Q(q;D_{n+1,s})}-m(D_{n+1,s})\right)\right\}=q\\
\Leftrightarrow&\frac{1}{\sqrt{v(D_{n+1,s})}}\left(\log\frac{Q(q;D_{n+1,s})}{1-Q(q;D_{n+1,s})}-m(D_{n+1,s})\right)=z_q\\
\Leftrightarrow&Q(q;D_{n+1,s})=\frac{\exp\{m(D_{n+1,s})+z_q\sqrt{v(D_{n+1,s})}\}}{1+\exp\{m(D_{n+1,s})+z_q\sqrt{v(D_{n+1,s})}\}}.
\end{align*}

\section{Metropolis-Hastings Algorithm and Approximation for Theorem 4.4}\label{append:MH}
The Metropolis-Hastings (MH) algorithm for generating random samples from $\boldsymbol{p}|\boldsymbol{Y}$ is given as follows.
\begin{algorithmic}[1]
  \For {$j=1$ to $J$}
	 \State Set $N_s=nT+s-1$.
     \State Start with a zero vector $\boldsymbol{p}$ of size $N_s$.
	 \For {$k=1$ to $N_s$}
          %\State Set $k=n(t-1)+i$. 
          \State Generate a random value $p^*_k$ from $Logitnormal(m(\boldsymbol{p}_{-k},\mathbf{y}_{-k}),v(\boldsymbol{p}_{-k},\mathbf{y}_{-k}))$.
          \State Generate an uniform random variable $U\sim Unif(0,1)$.
          \If {$U<\min\{1,\frac{f(y_{k}|p^*_{k})}{f(y_{k}|p_{k})}\}$}
          \State Set $\boldsymbol{p}=(p_1,\ldots,p^*_{k},\ldots,p_{N_s})$.
          \EndIf
     \EndFor	
	\State Set $\boldsymbol{p}^{(j)}=\boldsymbol{p}$
	\EndFor
	\State Return $\{\boldsymbol{p}^{(j)}\}_{j=1,\ldots,J}$.
\end{algorithmic}

In the algorithm, we first sample a value for the $k$-th component $p_k$ from the conditional distribution of $p_k$ given $p_j,y_j,j\neq k$, which is $Logitnormal(m(\boldsymbol{p}_{-k},\mathbf{y}_{-k}),v(\boldsymbol{p}_{-k},\mathbf{y}_{-k}))$, where 
$$m(\boldsymbol{p}_{-k},\mathbf{y}_{-k})=\mu_t(\mathbf{x}_i)-\sum_{k\neq j}\frac{Q_{kj}}{Q_{kk}} \left(\log\frac{p_k}{1-p_k}-\mu_t(\mathbf{x}_i)\right),\quad v(\boldsymbol{p}_{-k},\mathbf{y}_{-k})=\frac{\sigma^2}{Q_{kk}},$$
in which $\mu_t(\mathbf{x}_i)=\sum^R_{r=1}\varphi_ry_{i,t-r}+\mathbf{x}'_{i}\boldsymbol{\alpha}+\sum^L_{l=1}\boldsymbol{\gamma}_l\mathbf{x}_iy_{i,t-l}$ and $Q_{kj}$ is the $(k,j)$-element of $\boldsymbol{R}^{-1}_{\boldsymbol{\theta}}$.
Similar to \cite{zhang2002estimation}, we use the single-component MH algorithm, that is, to update only a single component at each iteration. Moreover, the proposed distribution $f(p_{k})$ is used for the single MH algorithm, so that the probability of accepting a new $p^*_{k}$ is the minimum of 1 and $\frac{f(p^*_{k}|y_{k})f(p_{k})}{f(p_{k}|y_{k})f(p^*_{k})}\left(=\frac{f(y_{k}|p^*_{k})}{f(y_{k}|p_{k})}\right)$. 

Based on the samples $\{\boldsymbol{p}^{(j)}\}_{j=1,\ldots,J}$, the mean, variance, and $q$-quantile of $p_s(\mathbf{x}_{n+1})|\boldsymbol{Y}$ can be respectively approximated by

$$\frac{1}{J}\sum^J_{j=1}\kappa(m(\boldsymbol{p}^{(j)},\boldsymbol{Y}),v(\boldsymbol{p}^{(j)},\boldsymbol{Y})),$$
\begin{align*}
\frac{1}{J}\sum^J_{j=1}\tau(&m(\boldsymbol{p}^{(j)},\boldsymbol{Y}),v(\boldsymbol{p}^{(j)},\boldsymbol{Y}))+\\
&\frac{1}{J-1}\sum^J_{j=1}\left[\kappa(m(\boldsymbol{p}^{(j)},\boldsymbol{Y}),v(\boldsymbol{p}^{(j)},\boldsymbol{Y}))^2-\frac{1}{J}\sum^J_{j=1}\kappa(m(\boldsymbol{p}^{(j)},\boldsymbol{Y}),v(\boldsymbol{p}^{(j)},\boldsymbol{Y}))\right],
\end{align*}
and the $q$-quantile of $\{p^{(j)}_s\}^J_{j=1}$, where  $p^{(j)}_s$ is generated from $Logitnormal(m(\boldsymbol{p}^{(j)},\boldsymbol{Y}),v(\boldsymbol{p}^{(j)},\boldsymbol{Y}))$. Similarly, the distribution of $y_s(\mathbf{x}_{n+1})|\boldsymbol{Y}$ can be approximated by the sample distribution of $\{y^{(j)}_s\}^J_{j=1}$, where $y^{(j)}_s$ is generated from a Bernoulli distribution with probability $p^{(j)}_s$.

\section{Algorithm: Dynamic Binary Emulator}\label{append:emulatecurve}

\begin{algorithmic}[1]
  \For {$j=1$ to $J$}
	 \State Set $N=nT$.
     \State Start with a zero vector $\boldsymbol{p}$ of size $N$.
     \For {$i=1$ to $N$}
	 %\For {$i=1$ to $n$; $t=1$ to $T$}
     %     \State Set $k=n(t-1)+i$. 
          \State Generate a random value $p^*_k$ from $Logitnormal(m(\boldsymbol{p}_{-k},\mathbf{y}_{-k}),v(\boldsymbol{p}_{-k},\mathbf{y}_{-k}))$.
          \State Generate an uniform random variable $U\sim Unif(0,1)$.
          \If {$U<\min\{1,\frac{f(y_{k}|p^*_{k})}{f(y_{k}|p_{k})}\}$}
          \State Set $\boldsymbol{p}=(p_1,\ldots,p^*_{k},\ldots,p_{N})$.
          \EndIf
     \EndFor	
     \State Set $\boldsymbol{p}_{n+1}=\boldsymbol{p},\boldsymbol{Y}_{n+1}=\boldsymbol{Y}$, zero vectors $\boldsymbol{p}_{\text{new}}$ and $\mathbf{y}_{\text{new}}$ of size $T$.
     \For {$t=1$ to $T$}
	 \State Given $D_{n+1,t}=\{\boldsymbol{p}_{n+1},\boldsymbol{Y}_{n+1}\}$, draw a sample $p_t(\mathbf{x}_{n+1})$ from $Logitnormal(m(D_{n+1,t}),v(D_{n+1,t}))$, and then draw a sample $y_t(\mathbf{x}_{n+1})$ from a Bernoulli distribution with parameter $p_t(\mathbf{x}_{n+1})$. 
	 \State Update $\boldsymbol{p}_{n+1}=(\boldsymbol{p}_{n+1}',p_t(\mathbf{x}_{n+1}))'$, $\boldsymbol{Y}_{n+1}=(\boldsymbol{Y}_{n+1}',y_t(\mathbf{x}_{n+1}))'$, $(\boldsymbol{p}_{\text{new}})_t=p_t(\mathbf{x}_{n+1})$, and $(\mathbf{y}_{\text{new}})_t=y_t(\mathbf{x}_{n+1})$.
	 \EndFor	
	      
	\State Set $\boldsymbol{p}_{\text{new}}^{(j)}=\boldsymbol{p}_{\text{new}}$ and $\mathbf{y}_{\text{new}}^{(j)}=\mathbf{y}_{\text{new}}$.
	\EndFor
	\State Take pointwise median from $\{\boldsymbol{p}_{\text{new}}^{(j)}\}_{j=1,\ldots,J}$ and $\{\mathbf{y}_{\text{new}}^{(j)}\}_{j=1,\ldots,J}$.
\end{algorithmic}
\end{appendices}

\bibliography{bib}

\end{document}